\documentclass[11pt,fleqn]{article}

\newcommand{\etal}{et al.\ }

\usepackage{amsthm, amssymb}
\usepackage{graphicx}
\usepackage{algorithm}
\usepackage{algorithmic}
\usepackage{xspace}
\usepackage{enumitem}
\usepackage{tabularx}
\usepackage{xcolor}
\usepackage{graphicx} 

\usepackage{tikz}
\usetikzlibrary{shadows}
\usetikzlibrary{patterns}
\usepackage{wrapfig}
\usepackage[fleqn]{amsmath}

\usepackage{thmtools} 

\usepackage{hyperref}
\hypersetup{colorlinks=true,citecolor=red}
\usepackage{cleveref}

\pgfdeclarepatternformonly{my crosshatch dots}{\pgfqpoint{-1pt}{-1pt}}{\pgfqpoint{1pt}{1pt}}{\pgfqpoint{6pt}{6pt}}%
{
    \pgfpathcircle{\pgfqpoint{0pt}{0pt}}{.5pt}
    \pgfpathcircle{\pgfqpoint{3pt}{3pt}}{.5pt}
    \pgfusepath{fill}
}

\newcommand{\R}{\mathbb{R}}
\newcommand{\Z}{\mathbb{Z}}
\newcommand{\Q}{\mathbb{Q}}

\newtheorem{theorem}{Theorem}[section]
\newtheorem{lemma}[theorem]{Lemma}

\newtheorem{observation}[theorem]{Observation}
\newtheorem{proposition}[theorem]{Proposition}
\newtheorem{corollary}[theorem]{Corollary}
\newtheorem{definition}[theorem]{Definition}
\newtheorem{claim}[theorem]{Claim}

\newtheorem{remark}[theorem]{Remark}

\newcommand{\set}[1]{\left\{#1\right\}}

\def\+#1{\mathcal{#1}}
\newcommand{\cI}{\mathcal{I}}

\newcommand{\cL}{{\mathcal L}}

\newcommand{\cP}{{\mathcal P}}

\newcommand{\cR}{{\mathcal R}}

\newcommand{\dd}{\text{d}}

\newcommand{\eps}{\epsilon}

\newcommand{\cost}{\mathsf{cost}}

\newcommand{\opt}{\mathsf{opt}}

\newcommand{\GD}{{\mathsf{GD}}}

\newcommand{\norm}[1]{\lVert#1\rVert}

\newcommand{\bc}{\textbf{c}}
\newcommand{\bb}{\textbf{b}}
\newcommand{\bx}{\textbf{x}}
\newcommand{\by}{\textbf{y}}
\newcommand{\bz}{\textbf{z}}
\newcommand{\bZ}{\textbf{Z}}

\newcommand{\vecf}{\textbf{f}}

\newcommand{\bp}{\textbf{p}}
\newcommand{\bq}{\textbf{q}}

\newcommand{\bone}{\textsf{1}}

\def\DEBUG{true}
\ifdefined\DEBUG
	
	\def\rem#1{{\marginpar{\raggedright\scriptsize #1}}}
	\newcommand{\sjr}[1]{\rem{\small\textcolor{red}{$\bullet${\tiny #1}}}}
	
	\newcommand{\remove}[1]{{\color{lightgray} #1}}
\else
	
	\newcommand{\sjr}[1]{}
	\new
	command{\remove}[1]{}
\fi
\def\DEBUG{true}
\ifdefined\DEBUG
	
	\def\rem#1{{\marginpar{\raggedright\scriptsize #1}}}
	\newcommand{\msr}[1]{\rem{\small\textcolor{red}{$\bullet${\tiny #1}}}}

\else
	
	\newcommand{\msr}[1]{}
	\new
	command{\remove}[1]{}
\fi

\begin{document}

\sloppy

\title{Online Scheduling via Gradient Descent  for Weighted Flow Time Minimization}


\author{
Qingyun Chen \thanks{ Electrical Engineering and Computer Science, University of California, 5200 N. Lake Road, Merced CA 95344. \texttt{qchen41@ucmerced.edu}.}   \and
Sungjin Im\thanks{ Electrical Engineering and Computer Science, University of California, 5200 N. Lake Road, Merced CA 95344. \texttt{sim3@ucmerced.edu}.}  
\and 
Aditya Petety\thanks{ Electrical Engineering and Computer Science, University of California, 5200 N. Lake Road, Merced CA 95344. \texttt{apetety@ucmerced.edu}.
} 
}

\date{}
\maketitle
\thispagestyle{empty}

\begin{abstract}

In this paper, we explore how a natural generalization of Shortest Remaining Processing Time (SRPT) can be a powerful \emph{meta-algorithm} for online scheduling. The meta-algorithm processes jobs to maximally reduce the objective of the corresponding offline scheduling problem of the remaining jobs: minimizing the total weighted completion time of them (the residual optimum). We show that it achieves scalability for minimizing total weighted flow time when the residual optimum exhibits \emph{supermodularity}. Scalability here means it is $O(1)$-competitive with an arbitrarily small speed augmentation advantage over the adversary, representing the best possible outcome achievable for various scheduling problems.

Thanks to this finding, our approach does not require the residual optimum to have a closed mathematical form. Consequently, we can obtain the schedule by solving a linear program, which makes our approach readily applicable to a rich body of applications. Furthermore, by establishing a novel connection to \emph{substitute valuations in Walrasian markets}, we show how to achieve supermodularity, thereby obtaining scalable algorithms for various scheduling problems, such as matroid scheduling, generalized network flow, and generalized arbitrary speed-up curves, etc., and this is the \emph{first} non-trivial or scalable algorithm for many of them.
\end{abstract}

\clearpage
\setcounter{page}{1}

\section{Introduction}
    \label{sec:intro}

Scheduling is a fundamental algorithmic problem because it is at the heart of any resource allocation process. Because of its wide appearance in practice, scheduling has been a central field in many disciplines such as operations research and computer science. Further, scheduling problems have provided arenas where theoretical computer science has developed and tested new theory. Scheduling theory has also become richer from various points of view, such as approximation, online, stochastic, game theory, etc. In the modern era,  scheduling is giving increasing challenges and opportunities as more powerful computing resources are becoming available at low cost for novel applications.

 Online scheduling studies dynamic scheduling settings and is found in various applications.	In the online setting, the algorithm must make scheduling decisions without knowing jobs arriving in the future. Competitive analysis is commonly used to measure the quality of an online algorithm in contrast to the optimal offline algorithm \cite{borodin2005online}. 	Unfortunately, in competitive analysis, online algorithms are often too weak compared to their offline adversaries that know the whole input sequence from the beginning. Particularly, the effect of online scheduling algorithms making sub-optimal decisions can accrue over time, often rendering their competitive analysis pointless.

To overcome this pessimistic landscape, in the mid-'90s, a resource augmentation model was introduced to online scheduling \cite{KalyanasundaramP00,phillips2002optimal}.  For example, if the augmented resource is speed, which is the most commonly augmented resource, the scheduling algorithm is compared to the adversary that can process jobs slightly slower. This beyond-worst-case analysis enabled meaningful study of a large number of scheduling problems, which would otherwise not be possible \cite{PruhsST04}. New analysis tools have been developed for online scheduling, such as potential functions \cite{ImMP11} and dual fitting \cite{AnandGK12,GuptaKP12}, which have resulted in a plethora of new exciting results 
for various online scheduling problems under the resource augmentation model.  

	\smallskip
	Despite all this success, our understanding of online scheduling can be considerably improved. First of all, we only have limited ways to design online scheduling algorithms systematically. For example, the powerful online primal-dual framework \cite{Niv-primaldual} is mainly applicable for monotone packing and covering problems but a large number of scheduling problems are not as such.  Dual fitting tries to analyze a candidate algorithm by setting the primal and dual LP variables according to the algorithm's behavior. Thus, designing online scheduling algorithms still requires much creativity, often non-trivially combining commonly used scheduling strategies via repeated trial and error. There is a need to find new algorithm design strategies that can keep pace with increasingly complicated scheduling environments.

Up to date, Proportional Fairness (PF) has been the only online meta scheduling algorithm known. The algorithm seeks to maximize the product of how fast jobs get processed at each time (subject to the scheduling constraints), and it generalizes Round-Robin.\footnote{In Round-Robin scheduling for a single machine, jobs are processed in turns after each unit of processing time. We assume that each unit of time is infinitesimally small. Under this assumption, the algorithm is equivalent to Processor Sharing, where all active jobs are processed at an equal rate at all times. In the single machine, PF solves $\max \prod_{j = 1}^n z_j$ s.t. $\sum_j z_j \leq 1$ and $\bz \geq 0$, thus yields the same assignment as Round-Robin.} It was only recently shown that PF is  $O(1/ \eps)$-competitive for a large class of scheduling problems, given $(e+\eps)$-speed over the adversary \cite{ImKM18}. Specifically, it was shown that if we view jobs as agents and how fast they get processed as their utility, PF finds a fair resource allocation that leads to a competitive schedule for the total flow time objective. The algorithm was not new, but the main contribution was to demonstrate its effectiveness as a meta scheduling algorithm---the guarantee holds true for various problems as long as the allocation is monotone under PF.\footnote{This means that when another job is added to the PF allocation, the share or speed allocated to each existing job decreases. } 

In this paper we aim to discover more meta scheduling algorithms that work for various scheduling environments.  
 We  focus on perhaps the most widely used metric, namely average (weighted) flow time, which measures the average time a job stays in the system from its arrival until completion. 

\medskip
\noindent 
\textbf{Roadmap.} 
Due to the general applicability of our approach, we will employ somewhat heavy notation. Further, we use concepts from market equilibria, which are not commonly used in the online scheduling/algorithm literature. To make the paper more accessible, we first describe our results informally (Section~\ref{sec:informal}) and then illustrate them using a single machine scheduling example (Section~\ref{sec:illustration}).
Readers unfamiliar with online scheduling are strongly recommended to \textbf{read Section~\ref{sec:illustration} before Section~\ref{sec:informal}}. There is some intentional overlap between these sections.
    
We then formally define our problem and introduce preliminary concepts about online scheduling and substitute valuations in Section~\ref{sec:prelim}. The linear substitute results and examples discussed in Section~\ref{sec:ls} are, to our best knowledge, new. We formally present our algorithms and theorems in Section~\ref{sec:formal}. Section~\ref{sec:main-proofs} is devoted to proving the main theorems. Section~\ref{sec:theorem-applications} demonstrates how these theorems are applied to the applications discussed in Section~\ref{sec:applications}. Other related work is discussed in Section~\ref{sec:other-related}. Appendix~\ref{sec:rt} discusses solving the residual optimum problem efficiently and implementing our 
 meta algorithm efficiently. All missing proofs are provided in the Appendix.

\subsection{Informal Description of Our Results and Contributions}
    \label{sec:informal}

We seek to find a meta scheduling algorithm that works for a large class of problems that fall into the polytope scheduling, which encompasses numerous scheduling problems \cite{ImKM18}.

\smallskip
\noindent
\textbf{Polytope Scheduling}.  
In the Polytope Scheduling Problem (PSP), we are given a set of constraints that form a downward-closed\footnote{If $\bf 0 \leq \bz' \leq \bz$ and $\bz \in \cP$, then $\bz' \in P$.} polytope $\cP$ (more generally a convex body), which constrains the feasible processing rate vector at each time: each job $j$ can be processed at a rate of $z_j$ if and only if $\bz := \{z_j\}_{j} \in \cP$.
For example, we have $\cP = \{\bz \; | \; \sum_j z_j \leq 1 \mbox{ and } \bz \geq 0\}$ in the single machine scheduling, which says jobs can be processed at a rate of up to 1 in total.  We assume that $\cP$ remains fixed over time. In the online setting, each job $j$ arrives at time $r_j$ with (original) processing size $p_j$ and weight $w_j$.  Note that the online scheduler learns about job $j$, including $p_j$ and $w_j$, only when it arrives. Preemption is allowed and the goal is to minimize total (weighted) flow time.

\smallskip

To develop a general approach to this problem, we get an inspiration from the Shortest Remaining Processing Time (SRPT): The SRPT algorithm, which prioritizes jobs with the shortest remaining time, is well-known for optimally minimizing the average (unweighted) flow time of jobs on a single machine. Interestingly, the scheduling literature implicitly suggests a possible interpretation of SRPT as a form of gradient descent \cite{ChadhaGKM09,AnandGK12,BansalCP13,BansalPS09}. 
The residual optimum \cite{ChadhaGKM09,AnandGK12}, informally, is the min cost for  completing all remaining jobs, assuming no more jobs arrive. Taking \emph{gradient descent on the residual objective}, meaning that we process the jobs that \emph{decrease the residual optimum the most}, is natural because we complete all the remaining jobs when the residual optimum becomes zero.
Further, the spirit of gradient descent is present in the analysis based on potential functions \cite{ChadhaGKM09,ImMP11,EdmondsP12}.

Inspired by this, we explore \emph{gradient descent as a meta algorithm for online scheduling}. Our approach is general and streamlined: i) Algorithm Execution: Solve an LP relaxation of the offline residual optimum and perform gradient descent; and ii) Performance Guarantee: If a valuation function coming from the scheduling constraints is a substitute function, gradient descent is competitive.

Below, we outline our results and contributions.

\begin{enumerate}
    \item \textbf{First Scalable Meta-algorithm.} We demonstrate the effectiveness of gradient descent as a \emph{meta-algorithm}. We use gradient descent to give the \emph{first scalable algorithm} for various scheduling problems, such as matroid scheduling, generalized flow scheduling, generalized arbitrary speed-up curves, etc.; see Section~\ref{sec:applications} for the problem definition. An algorithm is called scalable if it is $O(1_\eps)$-competitive with $(1+\eps)$ factor more speed given over the adversary.  Scalable algorithms are essentially the best one can hope for problems that have strong lower bounds on competitive ratio, and all the above problems are as such.    Since the aforementioned PF (Proportional Fairness) requires at least $(2+\eps)$-speed to be $O(1_\eps)$-competitive  \cite{Edmonds99,ImKM18}, our result gives the first scalable meta algorithm. We also note that we give the \emph{first non-trivial results} for generalized flow scheduling and generalized arbitrary speed-up curves.  For a comparison of our results to prior work, please see Table~\ref{table:summary}.

\end{enumerate}

\begin{table}
    \centering
    \begin{tabular}{|c|c|c|}
    
        \hline
        Problems & Previous Results& Our Results \\
        \hline\hline
        Matroid Scheduling & $e+\eps$ \cite{ImKM18} & $1+\eps$\\
        \hline
        Generalized Flow Scheduling & unknown  & $1 + \eps$\\
        \hline
        Unrelated Machines & $1 + \eps$ \cite{AnandGK12}& $1 + \eps$\\
        \hline
        Speed-up Curves on Partitioned Resources &  unknown  & $1 + \eps$ \\
        \hline
    \end{tabular}
    \caption{Comparison  of our results to prior work that summarizes the speed needed to be $O(1_\eps)$-competitive for minimizing total weighted flow time.  }
    \label{table:summary}
\end{table}

 As discussed earlier, the concept of gradient descent is  implicitly embedded within the online scheduling literature. In this paper, we present a general approach to using gradient descent and establish a sufficient condition that ensures its competitiveness. Rather than showcasing gradient descent's applicability to several scheduling problems by case-by-case analysis, we seek a general approach that can be used to solve a wide range of scheduling problems.

\begin{enumerate}
      \setcounter{enumi}{1}
    \item \textbf{Reduction to Obtaining an Offline (LP) Relaxation with Supermodularity.} We \emph{reduce} designing online scheduling algorithms to \emph{showing supermodularity of the offline residual optimum problem}. The residual optimum is the minimum weighted completion time objective to complete the remaining jobs, assuming no more jobs arrive.
    \begin{enumerate}
    
        \item
        The (approximate) residual optimum need not have a closed-form mathematical expression. Instead, it can be computed using mathematical programming, such as time-indexed linear programming (LP). This allows us to directly derive a gradient descent-based schedule. Consequently, our approach has a wide range of potential applications. In contrast, traditional potential function analysis relies on finding closed-form expressions, which has been a major hurdle in designing and analyzing online scheduling algorithms \cite{ImMP11}. Additionally, dual fitting is primarily an analysis tool, rather than a method for algorithm design.

        \item We show that \emph{supermodularity} of the residual optimum is a sufficient condition for gradient descent to be scalable. Intuitively, supermodularity implies adding a new job can only increase the residual optimum more when there are more existing jobs. Towards this end,  we generalize the beautiful potential function introduced by \cite{ChadhaGKM09} and crystallize the analysis by identifying supermodularity as the key to the analysis.             
    \end{enumerate}
\end{enumerate}

However, it is in general 
 challenging to show the residual optimum is supermodular. We address this issue by showing that it suffices to check the constraints that characterize the scheduling problem for a large class of problems. For example, 
 for scheduling $m$ identical parallel machines (we can process up to $m$ jobs at a rate of one at each time), we only need to 
 take a close look at the valuation function $v(\bx) = \max \sum_{0 \leq \by \leq \bx} \textbf{a} \cdot \by$ s.t. $|| \by ||_1 \leq m$ and $||\by||_\infty 
 \leq 1$ for some vector $\textbf{a} \geq 0$. 
 Intuitively, the valuation function means this: Given a set of jobs $J'$, what is the maximum value we can get by choosing $m$ jobs from $J'$ when each job $j$ has a value of $a_j$? Obviously, it is obtained by choosing $m$ jobs with the highest value of $a_j$. Alternatively, one can view the valuation as the weighted rank function of a uniform matroid of rank $m$.

    \begin{enumerate}
             \setcounter{enumi}{2}
   
    \item \textbf{Substitutes Empower Gradient Descent.} 
    We show that if maximizing a linear function subject to the constraints is a substitute function, then the residual optimum is supermodular. A valuation function is called \emph{gross/linear substitutes} (GS/LS) \cite{KelC82,GulS99,LehLN06,milgrom2009substitute} if increasing a resource's price does not reduce the demand for other resources. Consider the above valuation function $v(\bx)$ for $m$ identical machines: Items are individually priced at $\bq$ and valued at $\textbf{a}$ by you. If you buy $m$ items (say $1, 2, \ldots, m$) with the maximum value of $a_j - q_j$ (so, your utility is the valuation of the items you buy minus the payment for that), increasing item 1's price might make you buy a different item instead, but you would still want items $2,3, \ldots, m$. This shows that $v(\bx)$ is a gross substitute; linear substitute is a continuous extension of this concept.

    Substitute valuations play a crucial role in understanding market equilibria in Walrasian markets where agents aim to maximize their utility, defined as the difference between their valuation of purchased items and the price paid. Notably, Walrasian markets admit an equilibrium when all agents have GS valuations. An equilibrium is an allocation of items to agents alongside with a price vector where every agent gets an allocation maximizing their utility and all items are allocated. Furthermore, LP duality implies that this equilibrium maximizes the total valuation of all agents, also known as social surplus.

    We view the residual optimum problem (with a sign flip) as the social surplus maximization problem in Walrasian markets. Conceptually, we introduce a phantom agent for each time slot, where each job's unit processing is a distinct item.  If the valuation function is gross/linear substitutes (GS/LS), we prove that the residual optimum is supermodular. This leverages the fact that GS/LS valuations are a subclass of submodular functions and closed under max-convolution \cite{LehLN06,milgrom2009substitute}. This sufficient condition simplifies supermodularity checks. For example, in matroid scheduling, showing the valuation function (essentially the weighted rank) is GS/LS is easier than proving supermodularity for the entire residual optimum.

    Many scheduling scenarios involve jobs with processing rates dependent on allocated resources. To address this heterogeneity, linear substitutes (LS) \cite{milgrom2009substitute} become more relevant. Linear substitutes provide a natural continuous extension of GS. Less studied than GS, LS exhibits distinct behavior. We present novel results for LS, including Theorem~\ref{thm:gs-concave-closure}, which shows a concave closure of a GS valuation function is LS and Theorem~\ref{thm:gf-ls}, which shows generalized flow being LS. These findings hold independent interest. See Sections~\ref{sec:gs} and \ref{sec:ls} for more details and formal definitions. 
\end{enumerate}

Our work builds upon and extends numerous groundbreaking results that have been established throughout the history of scheduling research, e.g. \cite{HallSSW97,BansalPS09,ChadhaGKM09,EdmondsP12,AnandGK12,GuptaKP12,BansalCP13}. As a consequence, we demonstrate the efficacy of gradient descent as a meta-algorithm, which has perhaps been considered a possibility by many. This work perhaps suggests that SRPT, a common topic covered in algorithms and operating systems courses, should be taught alongside its connection to gradient descent.

\subsection{Overview of Our Approach}
\label{sec:illustration}

We illustrate our approach and its benefits using the single machine scheduling problem as a running example. 

\noindent 
\textbf{Single machine scheduling}. A set of $n$ jobs arrive over time to be processed on a single machine/processor and the scheduler is not aware of jobs that have not arrived yet. Job $j$ has a processing time (size) $p_j$ and arrives at time $r_j$. The scheduler is clairvoyant, meaning that it learns $p_j$ upon $j$'s arrival. The single machine can process at most one job or equivalently one unit of jobs at any point in time. Thus, letting $z_j$ denote the rate we process job $j$ at a time, feasible schedules are constrained by $\bz \in \cP = \{ \bz \geq 0 \; | \; \sum_j z_j \leq 1\}$.\footnote{If $\bz$ is fractional, we can implement it by preemptive time sharing.} A job $j$ is completed when it gets processed for $p_j$ units of time in total. Preemption is allowed, meaning that the scheduler can preempt a job being processed and resume it from where it left off. If we denote $j$'s completion time as $C_j$, its flow time\footnote{Flow time is also known as response time or waiting time.} is defined as $F_j := C_j - r_j$. Suppose the objective is to find an online schedule that minimizes total (or equivalently average) flow time, i.e., $\sum_j F_j$.

For this simple problem, as mentioned before, at each time $t$, SRPT (Shortest Remaining Processing Time) processes the job $j$ with the minimum $p_j(t)$, which denotes $j$'s remaining size at time $t$, that is, $p_j$ minus the amount of time for which it has been processed until time $t$. 
SRPT optimally minimizes the total flow time---intuitively, it minimizes the overall delay of job processing by processing the job it can complete the earliest. 

\subsubsection{Supermodularity Makes Gradient Descent Work (Theorem~\ref{thm:main-clairvoyant-int})}
\label{sec:overview-1}

Let us first interpret SRPT as a form of gradient descent (GD). Suppose the current time is $t$. Since we would like to minimize the total flow time, $\sum_j (C_j - r_j)$, by ignoring the cost of each job $j$ already incurred (how long it has waited), $t - r_j$, pretending the current time is 0, and assuming no more jobs arrive in the future, let us try to measure the minimum cost we should pay, which is termed the \emph{residual} optimum. The residual optimum then  becomes the minimum total completion time of jobs with sizes $\{p_j(t)\}_j$:
\begin{equation}
    \label{eqn:single}
f(\bp(t)) =   n p_{\pi(1)}(t)+ (n-1) p_{\pi(2)}(t) + (n-2) p_{\pi(3)}(t) + \cdots + 1 p_{\pi(n)}(t),
\end{equation}
where jobs are ordered so $\pi(1)$ has the smallest remaining size and $\pi(n)$ has the largest; we hide $\pi$'s dependency on $t$ for notational convenience. Note that computing $f(\bp(t))$ is a purely \emph{offline} problem. In general, $f(\bx)$ does not depend on a specific algorithm but only on the scheduling problem specified by the constraints, $\cP$.
Given the constraint that $\sum_j p_j(t)$ can be reduced by one at a time, GD processes job $\pi(1)$, 
because it decreases $f(\bp(t))$ the most.  This \emph{coincides with SRPT}. Thus, GD decreases the residual objective by exactly $n$, the number of jobs currently alive, at a time. Note that the algorithm designer  \emph{only needs to know the residual objective to derive GD}. While $f(\bp(t))$ has a nice closed form in the single machine case, it is often NP-hard to compute exactly. In such cases, we can use linear or convex programming (Sections~\ref{sec:illustration-2} and \ref{sec:Substitutes-Empower-GD}).

Let $\bp^O(t)$ denote the remaining job sizes under the adversary's schedule. Let $A$ denote GD algorithm (in this case, SRPT) and $O$ the adversary. Similarly, $A(t)$ and $O(t)$ imply the set of jobs currently alive at time $t$ in their schedule; we may drop $t$ for notational brevity.

For analysis, we use the following \emph{generic} potential function $\Phi(t)$, which is directly derived from $f$ in a \emph{black-box} manner: \vspace{-0.5mm}
$$\Phi(t) := \frac{2}{\eps}\left(f(\bp(t)) - \frac{1}{2}f(  \bp(t) \ || \ \bp^O(t)) \right)\vspace{-0.8mm},$$
where $\bp^O(t) := \{ \bp^O_j(t) \}_{j}$ denotes a vector of remaining job sizes in $O$'s schedule and $||$ denotes concatenation of two vectors. In other words, $f(  \bp(t) || \bp^O(t))$ computes the residual optimum pretending that each job exists in two separate copies, with remaining sizes $p_j(t)$ and $p^O_j(t)$, respectively. 
As mentioned before, this abstracts away the beautiful potential function used in \cite{ChadhaGKM09} for unrelated machines.\footnote{The potential function in \cite{ChadhaGKM09} looks somewhat different at first sight as it measures the aggregate increase of the residual optimum when each job in $A$ is added to $O$, and vice versa. But, by rearranging terms carefully, one can observe that our potential function is essentially equivalent to theirs in the setting of unrelated machines. } 

It is known \cite{ImMP11} that an online algorithm can be shown to be $O(1)$-competitive (with some speed augmentation) if the potential satisfies the following:
\begin{itemize}
\item (P1). $\Phi$ does not increase when a job arrives or is completed by $A$ or $O$.
\item (P2). $\frac{\dd}{\dd t} \Phi(t) \leq - |A(t)| + O(1) |O(t)|$ at all the other times.
\end{itemize}
This is because $\Phi(t  = 0) = \Phi(t = \infty) = 0$ (either no jobs have arrived or all jobs have been completed) and due to (P1) and (P2), we know that $0 \leq \int_{t = 0}^\infty \frac{\dd}{\dd t} \Phi(t) \dd t \leq  - \int_{t = 0}^\infty  |A(t)| \dd t+ O(1) \cdot \int_{t = 0}^\infty |O(t)| \dd t$
(the integral is taken over the whole time except when jobs arrive or are completed.) Note that $\int_{t = 0}^\infty  |A(t)| \dd t$ is $A$'s objective since every job alive at time $t$ incurs a unit cost at the time. Similarly, $\int_{t = 0}^\infty  |O(t)| \dd t$ is $O$'s objective. Thus we can establish $A$'s $O(1)$-competitiveness. 

Unlike \cite{ImMP11} that repeats trial-errors to find a potential function that satisfies these properties to analyze an online scheduling algorithm that a creative algorithm designer came up with, our approach \emph{reduces the whole algorithm design and analysis to checking if \textnormal{$f(\bx)$} satisfies supermodularity}. More formally, let $g(Y) := f(\bx \odot \bone_Y)$\footnote{The operation $\odot$ represents component-wise multiplication.} be the residual optimum only considering jobs in $Y$ when each job $j$'s (remaining) size is $x_j$.  Supermodularity  means 
$g(X \cup \{j\}) - g(X) \geq g(Y \cup \{j\}) - g(Y)$  for any $j$ and $X \supseteq Y$. In other words, it implies that the increase in the residual optimum caused by the addition of a new job can only be larger if there are more jobs.

We show that we can obtain (P1) from the supermodularity of $f$ by \emph{crystallizing} the analysis of \cite{ChadhaGKM09}. Further, (P2) follows from the nature of GD. In general, we can show that GD just follows the residual optimum schedule (until a new job arrives) and therefore decreases every alive job's projected completion time (counted from now) by one. Thus, the residual optimum (the minimum total completion time) decreases at a rate of $s |A(t)|$ when  GD is given speed $s > 1$, i.e., $\frac{\dd}{\dd t} f(\bp(t))  =  - s |A(t)|$.\footnote{Speed augmentation is a common beyond-worst-case model where the algorithm with slightly higher speed is compared to the adversary with speed 1.}
By the nature of GD, this is also the maximum possible decrease rate, which is the speed times the number of jobs alive. By thinking of the processing being done on $\bp(t) || \bp^O(t)$ as a feasible schedule with speed $s+1$ (combining GD's speed $s$ and the adversary's speed 1), we have 
$\frac{\dd}{\dd t} f(  \bp(t) || \bp^O(t)) \geq - (s+1) (|A(t)| + |O(t)|)$. Thus, (P2) follows  with $s = 1+\eps$. 

In summary, we demonstrated that GD can be naturally derived from the residual optimum and its competitiveness follows from the residual optimum's supermodularity. This argument extends to an arbitrary polytope scheduling problem (Theorem~\ref{thm:main-clairvoyant-int}).

\subsubsection{Substitutes Empower Gradient Descent (Theorems~\ref{thm:procesing-rate-view} and~\ref{thm:resource-view-2})}
    \label{sec:illustration-2}

Theorem~\ref{thm:main-clairvoyant-int} establishes supermodularity as a crucial structural property that ensures the effectiveness of GD. However, demonstrating supermodularity of the residual optimum is challenging for arbitrary problems, while trivial for the single-machine case. In fact, computing the residual optimum is NP-hard for many scheduling scenarios. 
Proving supermodularity without access to an efficient value oracle for the function  appears practically impossible.

To address these challenges, we can first consider an approximate residual optimum, which can be efficiently obtained by solving a time-indexed LP. Time-indexed LPs are widely used in the literature, e.g.  \cite{HallSSW97,AnandGK12}. We maintain our focus on single machine scheduling for illustration.
\begin{align*}
\min  \sum_{j, \tilde t \geq 1}  \frac{\tilde t}{p_j} z_{j \tilde t} \quad \mbox{ s.t. } \quad \sum_{j} z_{j \tilde t} \leq 1 \; \forall \tilde t \geq 1; \quad
\sum_{\tilde t \geq 1} z_{j \tilde t} \geq p_j(t) \; \forall j; \quad
\bz \geq 0 
\end{align*}
Here, we use $\tilde t$ in the LP to distinguish the time variables\footnote{Here, a unit time is assumed to be sufficiently small.} from actual time $t$. Variable $z_{j\tilde t}$ implies that job $j$ is processed at a rate of $z_{j\tilde t}$  at time $\tilde t$. The second constraint ensures that all jobs must be eventually completed. In general, the first constraint can be replaced with $\{z_{j \tilde t}\}_j \in \cP$, which encodes each scheduling problem's constraints. 
The objective is fractional since when job $j$ is processed at time $\tilde t$ by one unit, we pretend $\frac{1}{p_j}$ fraction of it is completed by the algorithm. It is well known that an algorithm that is scalable for the fractional objective can be  converted into one that is scalable for the integer objective 
(Lemma~\ref{thm:conversion}). 
Thus, it suffices to derive GD from the fractional residual optimum. We show that we can implement GD as follows: Solve the LP and follow the solution, i.e., 
process job $j$ at a rate of $z_{j\tilde t}$ at time $t + \tilde t$, until a new job arrives. We show how to make the LP compact in  Appendix~\ref{sec:rt} to make GD run in polynomial time. 
 
However, we now face another challenge: The optimum LP solution lacks a nice structure. Despite this, how can we show the supermodularity of the optimum LP objective as a set function of jobs? To show supermodularity, we make a connection to market equilibria. To this end, by flipping the sign (and adding some large constants $B$), we can convert the problem into a maximization problem as follows. Now the goal becomes showing \emph{submodularity} of the objective. 
\begin{align*}
\max  \sum_{\tilde t \geq 1} \sum_{j}  (B - \frac{\tilde t}{p_j}) z_{j \tilde t} \quad \mbox{ s.t. } \quad \sum_{j} z_{j \tilde t} \leq 1 \; \forall \tilde t \geq 1; \quad
\sum_{\tilde t \geq 1} z_{j \tilde t} \leq p_j(t) \; \forall j; \quad
\bz \geq 0 
\end{align*}

We consider the following (Walrasian) market: 
Items correspond to the remaining job sizes. Each item $j$ exists in $p_j(t)$ units. Each time $\tilde t$ has an agent with a valuation function,
\begin{align*}
v_{\tilde t}(\bx) := &\sum_{j}  (B - \frac{\tilde t}{p_j}) x_{j},
\end{align*}
where the function is defined over all $\bx \geq 0$ such that $||\bx||_1 \leq 1$.\footnote{Technically we need to extend the domain of $v_{\tilde t}$ to encompass all non-negative vectors $\bx \geq 0$. However, we will disregard this for the sake of illustration in this section.} Suppose the items are priced at $\bq$ (each unit of job $j$ is priced at $q_j$). To maximize her quasi-linear utility (valuation minus payment), the agent at time $\tilde t$ buys one unit of the item/job with the maximum value of $B - \tilde t / p_j - q_j$. 

A market equilibrium is an allocation of items to agents along with a price vector where every agent gets an assignment maximizing their utility and all items are sold out. The maximum total valuation of all agents is called the maximum social surplus, and it is well known that an \emph{equilibrium results in the maximum social surplus} from LP duality.

With this interpretation, the question becomes whether the max social surplus is submodular as a set function of items/jobs. Gross substitute (GS) valuations come to the rescue here. As mentioned before, a valuation is GS if increasing one item's price does not decrease the demand for other items in the agent's optimal choice. It is trivial to see $v_{\tilde t}$ is GS: The agent buys the item $j = \arg \max_{j'} B - \tilde t / p_{j'} -q_{j'}$ and increasing any other job's price does not affect her choice; the GS condition trivially holds if we increase  $j$'s price.

The beauty of GS lies in the fact that \emph{if every agent has a GS valuation, then the market equilibrium exists} and can be found by a tatonnement process where prices are monotonically increased from zero ensuring every item is demanded by at least one agent. In such a market, the agents with GS valuations collectively behave as a single agent with GS valuation. This aligns with the fact that GS functions are closed under max-convolution \cite{LehLN06}, resulting in \emph{the max social surplus being GS}.  Since \emph{GS is a subclass of submodular functions} (but not vice versa) \cite{GulS99,LehLN06}, we can conclude that the max social surplus is submodular as a set function of items/jobs, as desired.

\subsection{Applications}
    \label{sec:applications}

We outline some applications of GD.
We give scalable algorithms for all these applications---by using Theorem~\ref{thm:procesing-rate-view} for the first two and by using Theorem~\ref{thm:resource-view-2} for the third. The last application uses Theorems~\ref{thm:main-clairvoyant-int} and \ref{thm:procesing-rate-view} with a small tweak. 
The details, including how the scheduling problems are captured by the polytope scheduling (or equivalently multidimensional scheduling) are deferred to Section~\ref{sec:theorem-applications}.
\emph{No scalable} algorithms were known for any applications below except unrelated machines, prior to our work. In particular, \emph{no positive} algorithms existed for generalized flow scheduling and speed-up curves on partitioned resources.

\paragraph{Matroid Scheduling \cite{ImKM18,im2019matroid}.} In this problem, we are given a matroid $\+M = (J, \+I)$ where  $\+I$ is a collection of job sets. The pair $\+M = (J, \+I)$ is called a matroid if it satisfies: if $I' \subseteq I$ and $I \in \cI$, then $I' \in \cI$; and if $|I| < |I'|$ and $I, I' \in \cI$, then there exist $j \in I' \setminus I$ such that $I \cup \{j'\} \in \cI$. For numerous examples of matroids, see~Chapter 8 of \cite{Schrijver}.  At any time instant $t$, a feasible schedule is an independent set $I$ in $\cI$. If we schedule an independent set $I$ of jobs, then each job in $I$ gets processed at a rate of 1. We give a $(1+\eps)$-speed $O(1 / \eps^2)$-competitive algorithm. Previously, \cite{ImKM18} gave a $(e+\eps)$-speed $O(1 / \eps^2)$-competitive algorithm.

Matroid scheduling encompasses many scheduling problems. As discussed earlier, a uniform matroid of rank $m$ can model parallel identical machines (at every time step, we can process up to $m$ jobs). A gammoid captures single-source routing in a directed graph with a single source where jobs are processed when there are vertex-disjoint paths from the source to their corresponding nodes. A cographic matroid captures edge maintenance scheduling in an undirected graph. When performing service on an edge, it becomes unusable at the moment. However, concurrent maintenance work on some edges is possible if the remaining graph stays connected without those edges. In this scenario, maintenance jobs arrive at edges.

\paragraph{Generalized Flow Scheduling.} This generalizes the gammoid scheduling, a special case of the matroid scheduling. We are given a network $G=(V,E)$, where $J \subseteq V$ is the set of source nodes. Arc $e$ has capacity $u_e$ and flow gain factor $\gamma_e$; a flow of value $f$ becomes $\gamma_e f$ after flowing through $e$. Gain factors distinguish generalized flow from standard flow and are extremely useful for modeling \emph{heterogeneity}, as we can have a different factor for each edge.
Job $j$ gets processed at a rate equal to its (outgoing) net flow. It is critical for cloud computing to manage heterogeneous computing resources connected via a network; see \cite{jennings2015resource}. A job gets processed by communicating with the resources and the heterogeneous benefit between jobs and resources can be captured by edge gain factors, e.g. \cite{im2015spaa-tree}.  We give a $(1+\eps)$-speed $O(1 / \eps^2)$-competitive algorithm. 
No prior work exists for generalized flow scheduling. 

\paragraph*{Speed-up Curves on Partitioned Resources.} 
This generalizes the arbitrary speed-up curves model on uniform processors \cite{S0097539797315598,edmonds2003non,EdmondsP12}, which elegantly models how the parallelization effect degrades as we use more processors to speed-up processing a job. 
There are $D$ different divisible resources, where the $i$th resource exists in $m_i$ units. 

At every time, resources are allocated to jobs.
Each job has a valuation function of the following form: Resources are partitioned into $k$ disjoint groups $\+G = \set{G_1, ..., G_k}$; partitioning can be different for each job. Each group $G_i \in \+G$ is associated with a monotone univariate concave function $g_i: \R_+ \to \R_+$ with $g_i(0) = 0$. Then, the processing rate of job $j$ is $v_j(\bx_j) = \sum_{i=1}^{k} g_i(\sum_{d \in G_i} a_{jd} x_{jd})$,  where $a_{jd} \geq 0$, when $j$ receives  $\bx_j = (x_{j1}, x_{j2}, ..., x_{jD})$ resources. Concave functions are used to model the \emph{diminishing returns of parallelization} as we add more computing resources. The former arbitrary speed-up curves is a special case of this model when $D = 1$.  
We give a $(1+\eps)$-speed $O(1 / \eps^2)$-competitive algorithm. No prior work exists for this problem.

To see why this general model is useful, suppose $D = 2$ and consider three types of valuation functions:  $x_{j1}^{.5}$, $x_{j2}^{.75}$, and $(x_{j1}+ x_{j2})^{.5}$, for an example. The first means the job can only use processors in $G_1$, and the second means the job can only use those in  $G_2$, but the third can use any processors. This can model the \emph{restriction of processors} available for each job's processing. CPU partitioning is common in practice \cite{CPU-partitioning}.

\paragraph{Unrelated Machines \cite{lenstra1990approximation,sethuraman1999optimal,skutella2001convex,schulz2002scheduling,ChadhaGKM09,AnandGK12,Devanur014}.}  There is a set of parallel machines. Each job $j$ gets processed at a rate of  $\lambda_{ij}$ when scheduled on machine $i$; this is equivalent to $j$ having a processing time $p_{j} / \lambda_{ij}$ if scheduled only on machine $i$. Each job can be processed by at most one machine at any point in time. 

We first \emph{reproduce} the seminal result in \cite{AnandGK12}, which gave a $(1+\eps)$-speed $O(1/ \eps)$-competitive algorithm for minimizing total weighted flow time (Section~\ref{sec:unrelated-immediate}). We use the same algorithm of \cite{AnandGK12},  which immediately dispatches an arriving job to one of the machines and keeps the job on the same machine until its completion. Perhaps, the analysis of \cite{ChadhaGKM09,AnandGK12} becomes more transparent through the lens of our framework, supermodularity and gradient descent. 

While \cite{AnandGK12} gave essentially the best possible result and we can recover it, the work left an intriguing question from a technical point of view. Although immediate-dispatch and non-migration are highly desirable, they were strongly required in the analysis; it was not known whether we could still achieve a competitive algorithm if we migrated jobs in the middle of the execution for better load balancing. This has been puzzling in the research of online scheduling algorithms.

Our approach answers this question positively by providing a scalable algorithm that changes job assignments when a new job arrives to achieve better load balancing. Furthermore, our algorithm is purely \emph{Markovian}, meaning it only needs to remember the amount of jobs processed, eliminating the need to track job assignments. Specifically, we present a $(1+\eps)$-speed, $O(1 / \eps)$-competitive algorithm for unweighted total flow time (Section~\ref{sec:unrelated}) and a $(1+\eps)$-speed, $O(1 / \eps^3)$-competitive algorithm for weighted total flow time (Section~\ref{sec:unrelated-weighted}). While this result has a worse competitive ratio for the weighted case and is migratory, we believe the Markovian property has the potential for broader applications; see Section~\ref{sec:future}.

\subsection{Comparison with the Previous Work}
\label{sec:comparison}

As mentioned, Proportional Fairness (PF) \cite{bargaining,KMT,drf,JainV10} is the only online scheduling meta-algorithm known other than our gradient descent (GD). PF finds a fair allocation in Fisher markets, where each agent (job $j$) starts with an endowment (job $j$'s weight $w_j$) and buys some resources to maximize its valuation (how fast it gets processed at the moment). PF solves $\max  \sum_j w_j \log v_j(\by_j)$ s.t. $\sum_{j} \by_j \leq \bone$, $\by \geq 0$ where $j$ has valuation $v(\by_j)$ when it receives $\by_j$ resources from the supply vector of divisible resources.\footnote{Equivalently, the scheduling problem can be expressed as a polytope $\cP$ that describes how fast jobs can be processed, i.e., $\cP := \{ x_j = v_j(\by_j) \; | \; \sum_j \by_j \leq 1, \by \geq 0\}$.}  The solution can be viewed as an equilibrium where the market is cleared under the corresponding dual prices of the resources. Note that PF 
finds an instantaneous fair allocation without factoring in job sizes, hence is non-clairvoyant. The main contribution of \cite{ImKM18} lies in showing that PF is $O(1_\eps)$-competitive with $(e+ \eps)$-speed for the total weighted flow time objective for  scheduling problems when the allocation under PF is monotone---a job's speed/share can only decrease when other jobs arrive.

Although PF is a remarkable meta-algorithm generalizing RR (Round-Robin), it exhibits limitations from an algorithm design point of view. It is because we lose all the nice properties of PF when we deviate from the equilibrium it finds. Thus, it is not clear how to reduce the speed requirement of the algorithm; even RR requires $(2+\eps)$-speed over the adversary for competitiveness. 
Moreover, it disregards job sizes in scheduling decisions, a crucial aspect for designing machine learning-augmented algorithms that leverage job size estimates, a topic of recent significant interest \cite{purohit2018improving,im2021non,azar2021flow,AzarLT22,LindermayrM22}. 

Driven by these considerations, we investigate GD as a meta-algorithm, as the spirit of GD is often ingrained in online scheduling. Furthermore, unlike PF, GD offers flexibility. Implementing GD requires an estimate of the remaining cost. One common method of obtaining this estimate is to first fix a candidate algorithm and derive a reasonably accurate estimate of its remaining cost, which is converted into a potential function of a nice mathematical form \cite{ImMP11}. However, this approach demands algorithmic creativity, which involves a lot of trial and error. 

Alternatively, one can utilize the (approximately) optimal remaining cost, also known as the residual optimum. Note that the residual optimum is not tied to specific algorithms. However, computing the residual optimum is often NP-hard, rendering it challenging to understand its evolution. One approach is employing an LP relaxation, but this results in an LP solution, which commonly lacks structure.

Our work draws upon critical ideas from the seminal work of  \cite{ChadhaGKM09}, which presented the first scalable algorithm for unrelated machines; see~\ref    {sec:applications} for more details.  They introduced an elegant potential function, and our meta-potential function is its generalization. 

Generalizing their work poses challenges. Their algorithm executes SRPT (or its weighted version) on each machine following the initial assignment, and SRPT admits a closed-form expression for the residual optimum. Their potential function computes the residual optimum for a mixture of jobs from the algorithm and the adversary. However, determining how to proceed when the residual optimum lacks a convenient closed-form mathematical expression remains unclear. 

We address the challenges as follows. First, we abstract away the potential function argument presented in \cite{ChadhaGKM09} and demonstrate that their potential function can be generalized to work solely with supermodularity when gradient descent is employed. Furthermore, by establishing a novel connection to Walrasian markets, we simplify the task of verifying whether the residual optimum is supermodular.

The connection we make between gradient descent to Walrasian markets is less obvious than PF's connection to the Fisher markets. It is well-known that PF finds a market-clearing equilibrium in the Fisher market, but it is unclear how markets are related to the residual optimum. In Walrasian markets, agents have no endowments and they try to maximize their quasi-linear utility, which is their valuation minus the payment for the resources bought. We use gross substitutes\footnote{It is important to note that they have different definitions in Fisher and Walrasian markets.} and linear substitutes to find sufficient conditions to obtain supermodularity.

Finally, PF and GD are somewhat complementary to each other. 
PF makes scheduling decisions without knowing job sizes but uses $(e+\eps)$ speed to be $O(1)$-competitive. On the other hand, GD only needs $(1+\eps)$ speed augmentation. Further, there are problems for which one is $O(1)$-competitive (with $O(1)$-speed augmentation), but not the other. For example, for broadcast scheduling PF is $O(1)$-competitive, but not gradient descent (Appendix~\ref{sec:gd-bad-bdcast}). For unrelated machines, PF is believed not to be $O(1)$-competitive \cite{ImKM18}, but we show gradient descent is (Sections~\ref{sec:unrelated}, \ref{sec:unrelated-weighted}, and \ref{sec:unrelated-immediate}).

\subsection{Future Work}
    \label{sec:future}

In this paper, we show a reduction from designing online scheduling problems to the offline problem of finding an LP relaxation whose optimum satisfies supermodularity. In particular, we show if the underlying valuation function subject to the polytope characterizing the scheduling constraints at each time satisfies substitute properties, gradient descent is scalable. We discuss several open problems below. 

\begin{itemize}
    \item Gradient descent does not yield an $O(1)$-competitive algorithm for all scheduling problems as mentioned above (Appendix~\ref{sec:gd-bad-bdcast}).       
    Can we find other sufficient conditions for GD to be $O(1)$-competitive? For example, can we design algorithms for the intersection of two matroids?     
    The matroid intersection is known not to be  a gross substitute \cite{murota2022discrete}. Here, it may increase the required speed or  competitive ratio. 
    
    \item Can we extend gradient descent to $\ell_k$-norms of flow time? Although a naive extension of our framework to $\ell_k$-norms seems to create an undesirable increase of the potential due to job ages, we believe it should not be a major roadblock.  No meta-algorithms have been  considered for  $\ell_k$-norms of flow. We have some preliminary results for this objective. 
    
    \item Is there an online algorithm for minimizing $\ell_k$-norms of flow time where the competitive ratio has no dependency on $k$? For unrelated machines, there exist $(1+\eps)$-speed $O(k)$-competitive immediate-dispatch algorithms \cite{ImM11unrelated,AnandGK12}. However, it is known that any immediate dispatch algorithm must have a linear dependency on $k$ \cite{ImM11unrelated,AnandGK12}.  Since there exists a scalable algorithm for the maximum unweighted flow time with no dependency on $k$, it is plausible to remove the dependency by considering migratory algorithms \cite{anand2017maximum}. As discussed, our algorithm allows migration, and thus has the potential to solve this problem.

    \item Can we develop another nonclairvoyant meta-algorithm besides PF? PF generalizes Round Robin, which requires at least $(2+\eps)$-speed to be $O(1)$-competitive. However, Shortest Elapsed Time First (SETF)\footnote{In the single machine scheduling, SETF processes the job that has been processed the least.} is known to be scalable in the single-machine setting and uses information about how much jobs have been processed. Can we generalize SETF in a similar way that SRPT and Round Robin were generalized to gradient descent and PF?

\end{itemize}


\section{Preliminaries}
    \label{sec:prelim}

\subsection{Common Notation} We use $\bone_X$ to denote a binary vector corresponding to $X$. We use bold fonts for vectors. For two vectors $\bx, \by$, $\bx \wedge \by$ takes the component-wise min of $\bx$ and $\by$. Similarly, $\bx \vee \by$ takes the component-wise max of $\bx$ and $\by$. We let 
$\bx \odot \by$ denote the component-wise multiplication  of $\bx$ and $\by$. When $v$ is a set function, we will interchangeably use $v(X)$ and $v(\bone_X)$ for a set $X$.

\subsection{Problem Definition and Scheduling Primitives}
    \label{sec:problem-definition}

For easy reference, we reproduce the definition of the polytope scheduling below. All scheduling problems we will study fall into the polytope scheduling problem, or equivalently multidimensional scheduling. 

\paragraph{Polytope Scheduling (Feasible Processing Rate View).} Most preemptive scheduling problems with fixed computing resources can be represented as an instance of the polytope scheduling problem (PSP) \cite{ImKM18}. In the PSP, we are given a set of constraints that form a downward-closed\footnote{If $\bz' \leq \bz$ and $\bz \in \cP$, then $\bz' \in \cP$.} polytope $\cP$ (more generally a convex body), which encodes the feasible processing rate vector at each time: each job $j$ can be processed at a rate of $z_j$ if and only if $\bz := \{z_j\}_{j} \in \cP$.
For example, we have $\cP = \{\bz \; | \; \sum_j z_j \leq 1 \mbox{ and } \bz \geq 0\}$ in the single machine scheduling, which says jobs can be processed at a rate of up to 1 in total. 
We assume that $\cP$ remains fixed over time. 
\vspace{-3mm}
\paragraph{Multi-dimensional Scheduling (Resource View).} Equivalently, the PSP can be viewed as multidimensional scheduling where each dimension corresponds to a distinct resource. There exists a set of $D$ divisible resources or a continuous supply vector $\bone$,\footnote{We can assume wlog that every resource exists by one unit by scaling.} which are replenished at each time. Each job $j$ is associated with a concave valuation function, $u_j: [0, 1]^{D} \to \R_+$, which denotes how fast it gets processed. A feasible allocation is constrained by $\sum_j \by_j \leq \bone$: Job $j$ gets assigned a resource vector $\by_j$ and gets processed at a rate of $u_j(\by_j)$. We will call this resource view and the above feasible processing rate view. 

\medskip

We consider the clairvoyant scheduling setting: Each job $j$ arrives at time $r_j$ with (original) processing size $p_j$ and weight $w_j$. Note that the online scheduler learns about job $j$, including $w_j$ and $p_j$, only when it arrives. A job is completed when it has been processed by $p_j$ units since its arrival time $r_j$. Preemption is allowed, meaning a job can be paused and resumed later. Our goal is to minimize the following objective function.

\paragraph{Total Weighted Flow Time.} For integral flow time objective, an individual job $j$'s flow-time is defined as $F_j = C_j - r_j$ which is the difference between its arrival time $r_j$ and  completion time $C_j$; we may add $\sigma$ to $C_j$ as superscript to make clear the schedule $\sigma$ considered. Note that $C_j$ is defined the min time $t'$ such that $\int_{t = r_j}^{t'} z_j(t) \dd t = p_j$, where $z_j(t)$ is the processing rate of job $j$ at time $t$ in the schedule $\sigma$. 
Then, the total weighted integral flow time is $\sum_{j \in J} w_j F_j$, where $w_j$ is the weight of job $j$ and $J$ is the entire set of jobs arriving to be completed. 

\paragraph{Speed Augmentation.} A large number of scheduling problems do not admit $O(1)$-competitive algorithms. In such cases, speed augmentation is commonly used for beyond-worst-case analysis. Here, the online scheduling algorithm runs with extra speed compared to the optimal scheduler. If an algorithm can process $s >1$ times more than the adversary can process at a time and has an objective at most $c$ times than the adversary does, we say that the algorithm is $s$-speed $c$-competitive. 
Intuitively, this implies that the algorithm can handle $1 / s$-fraction of what the optimum solution can handle \cite{KalyanasundaramP00}. Thus, an algorithm that is $O(1_\eps)$-competitive with $(1 + \eps)$-speed for any $\eps > 0$ is of fundamental importance and termed  \textit{scalable}.

\paragraph{Total Fractional Weighted Flow Time.} In many scenarios, we will consider the fractional flow time which is a relaxation of integral flow time. The fractional weighted flow time objective is defined to be $\sum_{j \in J} \int_{t \geq 0} \frac{w_j}{p_j} t z_j(t) \dd t$. In other words, we view $\{z_j(t) / p_j\}_t$ as a distribution when the job is completed. By definition, it is obvious that a job's fractional flow time is no greater than its integer flow time in any fixed schedule. It is easier to minimize fractional flow time than integer flow time. Fortunately, with an extra speed augmentation, we can convert an algorithm that is good for fractional flow time into one that is good for integer flow time.

\begin{lemma}[\cite{ChadhaGKM09,ImMP11,ImKM18}]
    \label{thm:conversion}
    For any polytope scheduling problem, we can convert an online algorithm whose total fractional weighted flow time is $C$ into an online algorithm with (integer) weighted flow time $O(C / \eps)$ given $(1+\eps)$-speed. Therefore, if there is an online scheduling algorithm that is $c$-competitive with $s > 1$ speed for total fractional weighted flow time, then for any $\eps > 0$, there exists an algorithm that is $s(1+\eps)$-speed $O(c / \eps)$-competitive for total (integral) weighted flow time. 
\end{lemma}

For a fixed schedule $\sigma$, we will let $\sigma(t) := \{j \; | \; r_j \leq t < C^\sigma_j\}$  denote the set of jobs alive at time $t$ in $\sigma$. We will often use $A$ and $O$ to denote our algorithm('s schedule) and the adversary('s schedule), respectively. As an alive job $j$ adds $w_j$ at each time to the total weighted flow time objective, we immediately have the following. 

\begin{proposition}
    \label{pro:integral-weight}
    We have $\int_{t = 0}^\infty  W^\sigma(t) \dd t = \sum_j w_j (C^\sigma_j - r_j)$, where $W^\sigma(t) := \sum_{j \in \sigma(t)} w_j$ is the total weight of jobs alive at time $t$ in $\sigma$. 
\end{proposition}

This observation also extends to the fractional objective. Let $p^\sigma_j(t)$ be $j$'s remaining size at time $t$ in schedule $\sigma$.

\begin{proposition}
    \label{pro:fractional-weight}
    We have $\int_{t = 0}^\infty  \tilde W^\sigma(t) \dd t = \sum_j w_j \int_{t = r_j}^\infty  \frac{z_j^\sigma(t)}{p_j} t \dd t$, the total fractional weighted flow time of $\sigma$, where $\tilde W^\sigma(t) := \sum_{j \in \sigma(t)} w_j \frac{p_j^\sigma(t)}{p_j}$ is the total fractional weight of jobs alive at time $t$ in $\sigma$. 
\end{proposition}

\subsection{Gross Substitutes for Indivisible Resources}
    \label{sec:gs}

We first consider valuation functions for indivisible items.  Note that we consider the Walrasian market model throughout this paper, which is different from the Fisher market, as discussed in Section~\ref{sec:comparison}. Formally, the Walrasian market is defined as follows.

\begin{definition}[Walrasian Market]
    A Walrasian market consists of a set of $n$ indivisible items and $m$ agents where  each agent $i \in [m]$ has a valuation function $v_i: 2^{[n]} \to \R$. Each item $j \in [n]$ is priced at $q_j \in \R_+$. Then, the market gives a disjoint allocation $\{ S_i\}_{i\in[m]}$ of items over agents. We define the social surplus of the market as $\sum_{i \in[m]} v_i(S_i)$. 
\end{definition}

In the Walrasian market, each agent has a preference for bundles of items specified by the quasi-linear utility function, which can be formally defined by the demand correspondence as follows.

\begin{definition} [Demand Correspondence]
    Given a set function $v : 2^{[n]} \to \R$ and a price vector $\bq \in \R_+^n$, we define the demand correspondence (or set) $\+D(v, \bq)$ as the family of sets that maximize the quasi-linear utility of the agent:
    $$ \+D(v, \bq) := \arg\max_{X} \set{v(X) - \bq \cdot \bone_X}.$$
\end{definition}

In other words, $\+D(v, \bq)$ represents the set of item subsets that maximize our total happiness (the difference between our valuation and the total cost) under the price vector $\bq$, given our valuation function $v$. The Walrasian equilibrium asks that every agent maximizes its utility function and all items are allocated. Formally, the Walrasian equilibrium is defined as follows.

\begin{definition}[Walrasian Equilibrium]
    Given a Walrasian market with $n$ indivisible items and $m$ agents with valuations $\{v_i\}_{i \in [m]}$, a Walrasian equilibrium consists of a partition $\{ S_i\}_{i\in[m]}$ of $n$ and a price vector $\bq \in \R^n_+$ such that   
    $S_i \in \+D(v_i,\bq)$ for all $i \in [m]$. 
\end{definition}

It is known that an equilibrium exists if all agents' valuation function satisfies the GS property. GS is an important class of valuation functions \cite{KelC82}.

\begin{definition} [Gross Substitutes]
    A valuation function $v : 2^{[n]} \to \R$ is gross-substitute (GS) if for any price vectors $\bq' \geq \bq \geq 0$ and any $X \in \+D(v, \bq)$, there is a set $Y \in \+D(v, \bq')$ such that $X \cap \set{j : q_j = q'_j} \subseteq Y$.
\end{definition}

In other words, a GS valuation means that if we increase an item $j$'s price, we can update an optimum allocation without dropping any other items. Intuitively, this is possible because $j$ can be substituted with other items unlike some items behaving like a bundle, e.g. milk and cereal.

It is known that GS is a strict subclass of submodular functions \cite{GulS99,LehLN06}.

\begin{definition}
    \label{def:submodular}
    A valuation function $v: 2^{[n]} \rightarrow \R$ is submodular if for all $X, Y \subseteq [n]$, $v(X) + v(Y) \geq v(X \cap Y) + v(X \cup Y)$. Alternatively, it is submodular if $v(X \cup \{e\}) - v(X) \leq v(Y \cup \{e\}) - v(Y)$ for all $X \supseteq Y$ and $e \in [n]$.     
    Symmetrically, $v$ is supermodular if $-v$ is submodular. 
\end{definition}

\begin{lemma}[\cite{GulS99}]
    \label{lem:gs-submodular}
    A valuation function $v$ that satisfies the GS property is submodular.
\end{lemma}

Furthermore, Lehmann, Lehmann and Nisan \cite{LehLN06} showed GS 
valuations are closed under convolution, meaning that the aggregate max-valuation of GS valuations is GS.

\begin{lemma}
    \label{lem:GS-conv}
    The class of gross-substitute valuations is closed under convolution, i.e., for any two gross-substitute valuations $v_1$ and $v_2$, $v_1 \oplus v_2$ is gross substitutes, where 
    $$ v_1 \oplus v_2 (S) = \max_{T \subseteq S} \set{v_1(T) + v_2(S \setminus T)}$$
\end{lemma}

It is known that a Walrasian equilibrium exists when all agents have GS valuations, and a Walrasian equilbrium results in the max social surplus \cite{KelC82,GulS99}. Thus, the lemma can also be interpreted as stating that the social surplus of a Walrasian market exhibits GS if all agents have GS valuations.

\subsubsection{Examples}
    \label{sec:gs-examples}

We state a few examples of GS valuation functions. For a comprehensive list, see \cite{shioura2015gross}. 
\begin{itemize}
    \item Symmetric concave functions: $f: \set{0, 1}^{n} \rightarrow \R \cup \{ \infty \}$ is a symmetric concave function if there exists a concave function $\varphi: \R_+ \rightarrow \R$ s.t $f(\bx) = \varphi \left(\sum_{i=1}^{n}x_i \right)$. 
    
    \item Laminar concave functions: A set family $\mathcal{T} \subseteq 2^{N}$ is laminar if $X \cap Y = \emptyset$, $X \subseteq Y$ or $X \supseteq Y$ for every $X \neq Y \subseteq \mathcal{T}$. Define $\varphi_{Y}: \R_{+} \rightarrow \R$ to be a univariate concave function. Then the function given by $f(x) = \sum_{Y \in \mathcal{T}} \varphi_{Y}\left(\sum_{i \in Y} x_i \right)$ is called a laminar concave function.
    
    \item Maximum weight bipartite matching: Let $G=(U, W; E)$ be a weighted bipartite graph with an edge cost function $w : E \to \R$. Define  $f : \set{0, 1}^n \to \R \cup \set{-\infty}$ to be the maximum weight bipartite matching w.r.t a vertex set $X \subseteq U$ by
    $$f(\bone_{X}) = \max \Big\{\sum_{e \in M} w(e) \ | \ \text{$M \subseteq E$ is a matching in $G$ covers all vertices in $X$} \Big\}$$

    \item Weighted rank functions of matroids: Let $\+M = (E, \mathcal{I})$ be a matroid (see Section~\ref{sec:applications} for definition), where $E$ is a universe of $n$ elements, and $\+I \subseteq 2^E$ is a collection of independent sets. Let $w : E \to \R_+$ be a weight function. Then define $f: \set{0, 1}^n \to \R_+$ to be the weighted rank function w.r.t an element set $X \subseteq E$, i.e., $$f(\bone_{X}) = \max\Big\{\sum_{i \in Y}w(i) \ | \ Y \subseteq X, Y \in \mathcal{I}\Big\}$$

\end{itemize}

\subsection{Linear Substitutes for Divisible Resources}
\label{sec:ls}

Since we consider preemptive schedules where jobs can be processed by fractional amounts at a time and we will handle heterogeneous resources, we need a continuous extension of GS and Walrasian markets where resources to be allocated are divisible. 
Continuous definitions of gross substitutes, termed linear substitutes (LS),  were introduced by Milgrom and Strulovici \cite{milgrom2009substitute}.
 Interestingly, although the continuous setting does not inherit all properties that are guaranteed in the discrete setting \cite{milgrom2009substitute,shioura2015gross}, LS still offers some useful properties we can leverage. In particular, LS valuations also are closed under convolution (or equivalently called aggregation) and are a subclass of submodular functions.\footnote{For a continuous function $v$, submodularity is defined analogously to Definition~\ref{def:submodular}: 
  $v(\bx) + v(\by) \geq v(\bx \vee \by) + v(\bx \wedge \by)$. 
 
  Further, $v$ is supermodular iff $-v$ is submodular.}

\begin{definition} [Linear Substitutes]
    $v: \prod_{i \in [n]} [0, L_i] \rightarrow \R$ is a linear-substitute valuation if for any price vector $\bq$ and $\bq'$, then whenever $q_j \leq q'_j$, $q_k = q'_k$ for all $k \neq j$ and $X \in \+D(v, \bq)$, there is a $Y \in \+D(v, \bq')$ such that $Y_k \geq X_k$ for all $k \neq j$. 
\end{definition}

\begin{lemma}[\cite{milgrom2009substitute}]
    \label{lem:LS-aggregation}
    The class of linear-substitute valuations is closed under aggregation, i.e., for any linear-substitute valuations $v_1, ..., v_n$, $v$ is linear substitutes, where 
    $$ v(\bx) = \max \set{\sum_{i \in [n]}v_i(\bx_i) : \sum_{i \in [n]} \bx_i = \bx, \bx_i \geq 0 \  \forall i \in [n]}$$
\end{lemma}

Similar to GS, this lemma implies that the social surplus is LS if all the agents have LS valuations in the Walrasian market considered. 

To verify the LS property, we can check the submodularity of its dual profit function.

\begin{lemma}[\cite{milgrom2009substitute}]
    \label{lem: LS-equal}
    $v : \R^{n} \to \R$ is a linear-substitute valuation iff $\pi$ is submodular, where 
    $$ \pi(\bq) := \max_{\bx \in \R^n} v(\bx) - \bq \cdot \bx $$
\end{lemma}

We can easily extend the proof of GS implying submodularity (Lemma 5 in \cite{GulS99})  to LS implying submodularity; for completeness, we give the proof in Appendix~\ref{sec:ls-proofs}.

\begin{lemma}
    \label{lem:LS-submodular}
    Let $v : \R^n \to \R$ be a concave linear-substitute valuation. Then, $v(\bx)$ is submodular.
\end{lemma}

\subsubsection{Examples} \label{section:examples}

We discuss some examples of linear substitutes, which will be used as building blocks in applying our theorems. Compared to gross substitutes, relatively fewer examples are known in the literature. We show that the concave closure of GS and generalized flow satisfy LS, which could be of independent interest; the proofs are deferred to Appendix~\ref{sec:ls-proofs}. 
For more examples, see Section 5 of \cite{murota2022discrete}.

\begin{restatable}[Concave Closure of Gross Substitutes is LS]{theorem}{ThmGsClosure}
    \label{thm:gs-concave-closure}
    Given a gross-substitute valuation $v : \set{0, 1}^n \to \R$, let $ v^+(\bx) := \max \set{\sum_S v(S) \lambda_S : \sum_{S} \lambda_S \bone_S = \bx, \sum_{S} \lambda_S = 1, \lambda_S \geq 0}$ be the concave closure of $v$. Then, $v^+(\bx): [0, 1]^n  \to \R_+$ is linear substitute.    
\end{restatable}

To keep the flow of presentation, we defer the formal definition of generalized flow to Section~\ref{sec:theorem-applications}.

\begin{theorem}[Generalized Flow is LS]
    \label{thm:gf-ls}
    For any cost vector $\bc$, $v(\bx) := \max_{\bz \in \cP} \bc \cdot \bz$ for the $\cP$ in Eqn.~(\ref{eqn:p-gf}), which is defined for generalized flow, 
   is LS.
\end{theorem}

\section{Formal Description of Our Results}
    \label{sec:formal}

\subsection{Gradient Descent on the Residual Objective}
\label{sec:GD-on-Res}

We begin by formally defining the residual objective at each time $t$, namely the minimum projected cost of the remaining schedule. Let $A(t) := \{ j \; | \; r_j \leq t < C_j\}$ be the set of jobs alive at time $t$ in our schedule. Let $p_j(t)$ denote $j$'s remaining size at time $t$; by definition, $p_j(t) = 0$ for all $t \geq C_j$. Since the objective we consider is a linear combination of jobs' flow time, we can observe that $w_j (C_j - t)$, is exactly how much $j$ will contribute to the objective from time $t$ until it is completed at time $C_j$ as  $j$ gets processed by $p_j(t)$ units  by the algorithm  during $[t, C_j]$. So, if we shift time by $t$ for easy indexing, we have the following residual scheduling problem: Each alive job that has arrived has a (remaining) size $p_j(t)$ and has arrival time $0$ (in the shifted time horizon), and weight $w_j$. Noticing that completion time is equivalent to flow time when jobs arrive at time 0, the residual optimum can be expressed as follows.
 
\begin{align} 
\mbox{Residual Optimization Problem $\cI^{\cP}_\bx$:} \quad \quad  f(\bx) :=   \min \sum_{j \in A(t)} & w_j \tilde C_j  \quad s.t.  \label{residual-IP}\\
     \int_{\tilde t \geq 0}^{\tilde C_j} z_j(\tilde t) \dd \tilde t &= x_j \quad  \forall j \in A(t) \nonumber \\
     \bz(\tilde t) &\in \cP \quad  \forall \tilde t \geq 0, \nonumber 
\end{align}
where we set $x_j = p_j(t)$ for all $j \in A(t)$.  Note that we added tilde to $C_j$ and $t$ to make it clear that we are considering the shifted time horizon. The first constraint ensures that every job $j$ must be processed by $x_j$ units, and the second guarantees feasible scheduling at every time.

Let $\cI^\cP_\bx(\bz)$ denote the objective of $\cI^\cP_\bx$ for $\bz$. In other words, $\cI^\cP_\bx(\bz)$ is the extra cost we should pay if we follow $\bz$ from time $t$, i.e., process each job $j$ at a rate of $\bz_j(\tilde t)$ at time $t + \tilde t$, assuming no more jobs arrive.  If $\bz$ is infeasible, then let $\cI^\cP_\bx(\bz) = \infty$. Let $f(\bx) := \min_{\bz} \cI^\cP_\bx(\bz)$, which we will term algorithm $A$'s (integral) residual optimum at time $t$. 

We are now ready to formally define our gradient descent algorithm: At each
time $t$ we try to decrease $f(\bx)$ the most, i.e.,
\begin{align} 
\mbox{Gradient Descent on the Residual Optimum $f(\bx):$ } \quad \quad   \min_\bz \nabla f(\bx) \;|_{\bx = \bp(t)}  &\cdot \bz \label{eqn:gd}\\
         \bz &\in \cP \nonumber
 \end{align}

We will let $\GD^\cP(\bx)$ denote the decrease of the integral residual optimum due to the gradient descent, i.e., the optimum to Eqn. (\ref{eqn:gd}).   However, we will see that if $\bz(\tilde t)$ for $\tilde t \geq 0$ is an optimum solution that achieves $f(\bp(t))$, to implement GD, we can simply process jobs according to $\bz(0)$ at time $t$. In fact, we will show that we can exactly follow the residual optimum schedule until a new job arrives (Lemma~\ref{lem:GD-redisual-2}).  For brevity, we may omit $\cP$ from notation.

Notably, we take this time-indexed view since it explicitly describes the residual optimum solution at each time. Let us illustrate this with a single machine scheduling example we discussed in Section~\ref{sec:illustration}: After solving $\cI^{\cP}_{\bp(t)}$ we obtain: $z_{\pi(j)}(\tilde t) = 1$ for $\tilde t \in [0, p_{\pi(1)}(t))$ only for $j = 1$ (0 for all other jobs), $z_{\pi(j)}(\tilde t) = 1$ for $\tilde t \in [p_{\pi(1)}(t), p_{\pi(1)}(t)+ p_{\pi(2)}(t))$ only for $j = 2$, and so forth. By solving~Eqn.(\ref{eqn:gd}), we obtain $\bz$ such that $z_{\pi(1)} = 1$ and $z_{\pi(j)} = 0$ for all $j \neq 1$. Note 
that $\bz$ is identical to $\bz(0)$. This is not a coincidence. Lemma~\ref{lem:GD-redisual-2} says that we can follow $\bz(\tilde t)$ to implement GD: process job $\pi(1)$ during time 
$[t, t+ p_{\pi(1)}(t))$ and job $\pi(2)$ during $[t+p_{\pi(1)}(t), t+ p_{\pi(1)}(t)+ p_{\pi(2)}(t))$, and so forth. This is exactly what SRPT does. 

In the following, we define the key desired property we use in our paper: for any fixed remaining job sizes $\bx$, the residual optimum viewed as a set function is supermodular. There is an extended notion of submodularity/supermodularity in the continuous setting which is needed for linear substitutes (Section~\ref{sec:ls}). However, this notion of supermodularity is sufficient for our purpose here. For brevity, we may simply say supermodularity without explicitly mentioning `discrete.'
\begin{definition}
    \label{def:d-super}
    We say a function $g: \R^{|J|}_{\geq 0} \rightarrow \R$ is 

        (discrete-)supermodular if for any fixed $\bx$,  $\bar g_{\bx}(Z) := g(\bx \odot \bone_Z): 2^{J} \rightarrow \R$ is supermodular, i.e., $\bar g(U) + \bar g(V) \leq \bar g(U \cap V) + \bar g(U \cup V)$ for all $U, V \subseteq J$, where 
        $\bone_Z$ is the binary characteristic vector of $Z$ for any $Z \subseteq J$ and 
        $\odot$ denotes component-wise multiplication.
\end{definition}
We identify supermodularity as a key sufficient condition for the gradient descent on the optimum integral residual objective to give a scalable online algorithm by the following theorem. 
\begin{restatable}[Gradient Descent Desiderata for Integral Objective]{theorem}{thmone}
\label{thm:main-clairvoyant-int}
 For a polytope scheduling problem with polytope $\cP$,  
 suppose $f(\bx) := \min_{\bz} \cI^\cP_\bx(\bz)$ is discrete-supermodular.
     Then, for any $\eps > 0$, GD is $(1+\eps)$-speed $O(1 / \eps)$-competitive for minimizing weighted integral flow time for PSP $\cP$. 
\end{restatable}

\subsection{Substitutes Empower Gradient Descent}
    \label{sec:Substitutes-Empower-GD}

Unfortunately, a large number of scheduling problems are NP-hard to optimize for the integral completion objective.

Further, it is non-trivial to check if the residual optimum objective is supermodular. Thus, to make the gradient descent framework more readily applicable, we consider a fractional version of the residual objective and show it satisfies supermodularity for a large class of scheduling problems related to substitute valuation functions. Recall that we can use a standard conversion to translate an algorithm that is scalable for fractional weighted flow into one that is scalable for integer weighted flow; see Lemma~\ref{thm:conversion}. 

As discussed in Section~\ref{sec:problem-definition}, we can take the feasible processing rate view (polytope scheduling) or equivalently the resource view (multi-dimensional scheduling). In the following, we consider both views.

\subsubsection{Feasible Processing Rate View}

The following is a well-known time-indexed LP \cite{hall1997scheduling}, which we will use to measure our fractional residual objective. As before, we use $\tilde t$ to distinguish the time used in the residual schedule from the global time $t$.
\begin{align}
\mbox{Residual (Time-Indexed) LP $\cL^{\cP}_\bx$:} \quad \quad    
     f(\bx):= \min \sum_{j \in A(t)} & \frac{w_j}{p_j} \int_{\tilde t \geq 0} \tilde t \cdot z_j(\tilde t) \dd \tilde t & \label{eqn:frac-residual} \\
     \int_{\tilde t \geq 0} z_j(\tilde t) \dd \tilde t &= x_j \quad  \forall j \in A(t) \nonumber \\
     \bz(\tilde t) &\in \cP \quad  \forall \tilde t \geq 0 \nonumber 
 \end{align}

This LP tries to minimize total weighted fractional completion time. 
Gradient descent on $f(\bx) := \min_{\bz} \cL^\cP_\bx(\bz)$ solves $\min \nabla f(\bx) \;|_{\bx = \bp(t)} \cdot \bz$ subject to $\bz \in \cP$ and processes jobs at time $t + \tilde t$ according to $\bz(\tilde t)$, but as before, we will show that we can follow the LP solution until a new job arrives. 

\begin{restatable}[Gradient Descent Works for Linear Substitutes in Feasible Processing Rate View]{theorem}{thmtwo}
    \label{thm:procesing-rate-view}
    For a polytope scheduling problem with polytope $\cP$, suppose that for any $\bc \in \R_+^{|J|}$, $v(\bx) = \max_{0 \leq \bz \leq \bx, \bz \in \cP} \bc \cdot \bz$ is linear substitute. Then, for any $\eps > 0$, the GD gives a $(1 + \eps)$-speed $O(1 / \eps)$-competitive online algorithm for minimizing total fractional weighted flow time; thus, there exists a $(1+\eps)$-speed $O(1 / \eps^2)$-competitive algorithm for total weighted flow time.  
\end{restatable}

The theorem can be proved separately for gross substitutes (GS) for indivisible resources. However, we only state the theorem for linear substitutes (LS) because we can show that the concave closure of GS is LS (Theorem~\ref{thm:gs-concave-closure}). The conceptual contribution of the theorem lies in connecting GD to substitute valuation.

The time-indexed LP could be exponentially large. 
In Appendix~\ref{sec:rt}, we show that we can consider an approximate residual optimum to obtain a polynomial time 
algorithm---assuming that we have an efficient oracle to test whether $\bz' \in \cP$ or not.

\subsubsection{Resource View} 

We now consider the other equivalent view where each job gets some resources at each time from a $D$-dimensional supply resource vector $\bone$, which is replenished at every time. Recall that each job $j$ is associated with a valuation function $u_j$. Let $\by_{jt}$ denote the resource vector job $j$ receives at time $t$. Job $j$ then gets processed at a rate of $u_j(\by_{jt})$. We will keep reserving $\tilde t$ to refer to times in the residual schedule. The LP we solve is essentially identical to $\cL^{\cP}_\bx$, but has more explicit constraints that resources cannot be over-allocated. Concretely, we use the following linear programming.
\begin{align}
\mbox{(Resource View) Residual LP (or CP) $\cR_\bx$:} \quad \quad    
     f(\bx) := \min \sum_{j} & \frac{w_j}{p_j} \int_{\tilde t \geq 0} \tilde t \cdot u_j(\by_{j\tilde t}) \dd \tilde t & \label{resource-view-ip} \\
     \int_{\tilde t \geq 0} u_{j}(\by_{j\tilde t}) \dd \tilde t &= x_j \quad  \forall j \nonumber \\
     \sum_{j} \by_{j \tilde t} &\leq \bone \quad \forall \tilde t \geq 0 \nonumber 
 \end{align}

We prove the following theorem.

\begin{restatable}[Gradient Descent Works for Linear Substitutes in Resource View]{theorem}{thmthree}
    \label{thm:resource-view-2}
    For a multi-dimensional scheduling problem with a supply vector $\bone$ of divisible resources, suppose that for every job $j$, $u_j : [0, 1]^{D} \to \R_+$ is strictly monotone and strictly concave\footnote{Under this assumption, the demand set is a singleton set.} linear substitute  and further satisfies the following:
    \begin{enumerate}
        \item For any $\by \in \+D(u_j, \bb)$ and $\lambda \geq 1$, there exists $\by' \in \+D(u_j, \lambda \bb)$ such that $\by \geq \by'$.
        \item For any $\by \in \+D(u_j, \bb)$ and $\bb' \leq \bb$, there exists $\by' \in \+D(u_j, \bb')$ such that $u_j(\by') \geq u_j(\by)$.          
    \end{enumerate}    
    Then, for any $\eps > 0$, the GD applied to $\cR_{\bp(t)}$
    yields a $(1 + \eps)$-speed $O(1 / \eps)$-competitive online algorithm for minimizing total fractional weighted flow time; thus, there exists a $(1+\eps)$-speed $O(1 / \eps^2)$-competitive algorithm for integer weigthed flow time. 
\end{restatable}

In the theorem, in addition to every job $j$ having an LS valuation function, we require two additional properties. The first says that if we increase all prices by the same factor, then each job demands only less resources. The second says that if we increase resource prices, each job's optimal demand maximizing its quasi-linear utility should have a lower valuation. While these two properties are intuitive, not all LS valuations satisfy them\footnote{For example,  $u(\by) := \max \sum_i \lambda_i y'_i \ \text{s.t.} \ ||\by'||_1 \leq 1$ and $||\by'||_1 \leq ||\by||_1$ is LS, but does not satisfy the first property.}.

To prove the theorem, we show that the valuation functions considered here induce an LS function in the feasible processing rate view, which allows us to use Theorem~\ref{thm:procesing-rate-view}. The high-level proof idea is as follows. Consider $v(\bx) = \max_{0 \leq \bz \leq \bx, \bz \in \cP} \bc \cdot \bz$ for a fixed $\bc$. 
We want to show the valuation $v$ is LS. Let $\bq$ be the job prices per unit. 
By processing job $j$ by one unit, we obtain value $c_j$, yet have to pay the price $q_j$. Finding the demand set for $v$  can be viewed as finding the maximum social surplus when each job (agent) $j$ has valuation $(c_j - q_j)u_j$ in the resource view; let $\bb$ be the equilibrium prices of the $D$ resources. 
Increasing $q_j$ lowers $(c_j - q_j)u_j$, which has the same effect of increasing the resource prices $\bb$ for $j$ by a certain factor $\lambda > 1$.  Thus it results in resources being under-utilized due to the first property. Then, it is known that the  equilibrium prices drop \cite{ausubel2004auctioning, milgrom2009substitute}. All jobs' valuation functions $(c_i - q_i) u_i$, except for $j$, remain the same. Therefore, they can only get processed more due to the second property.

\section{Proofs of Main Theorems}
    \label{sec:main-proofs}

\subsection{Proof of Theorem~\ref{thm:main-clairvoyant-int}}
    \label{sec:proof-thm:main-clairvoyant-int}

It is highly important to note that in the proof when a job $j$ arrives, we pretend that GD and the optimum schedule each are given a \emph{distinct} copy of $j$. Therefore, if $A(t)$ and $O(t)$ denote the sets of jobs that have arrived yet have not been completed at time $t$ by GD and by the optimum schedule respectively, then $A(t) \cap O(t) = \emptyset$ always. We will use $A$ and $O$ to denote GD's schedule and the optimum schedule, respectively.  

We use $\bx || \by$ to denote the concatenation of two vectors $\bx$ and $\by$. Let $p^A_j(t)$ and $p^O_j(t)$ denote the remaining size of job $j$ under our algorithm GD and a fixed optimum schedule respectively. For brevity we let $\bp^A(t) := \{p^A_j(t)\}_{j \in A(t)}$ and $\bp^O(t) := \{p^O_j(t)\}_{j \in O(t)}$. For analysis, we will use the following potential function. 

$$ \Phi(t) = \frac{2}{\eps}\left(f(\bp^A(t)) - \frac{1}{2}f(\bp^A(t) || \bp^O(t))\right)$$

Following the standard potential function argument in online scheduling \cite{ImMP11} we will show the following:

\begin{enumerate}
    \item Boundary Conditions: $\Phi(0) = \Phi(\infty) = 0$. In other words, at the beginning when no jobs have arrived, and at the end when all jobs have been completed, the potential is 0. This is obvious because the residual optimum is 0 when there are no jobs to be processed and it is only concerned with processing jobs that are currently alive. Further, it's worth noting that jobs of zero remaining sizes have no effect on the residual optimum.
    
    \item Discrete Changes: $\Phi$ does not increase when a job arrives or is completed by GD or the optimum schedule. To this end, we will crucially use supermodularity.
    \item Continuous Changes: $\frac{\dd \Phi(t)}{\dd t} \leq - W^A(t) + O(\frac{1}{\eps}) W^O(t)$, when $A$ is given $1+\eps$ speed and $O$ 1-speed, where $W^A(t) : = \sum_{j \in A(t)} w_j$ and $W^O(t) := \sum_{j \in O(t)} w_j$ denote the total weight of jobs alive at time $t$ in GD's schedule and the optimum's, respectively. 
\end{enumerate}

\begin{figure}[htbp]
    \centering
    \begin{tikzpicture}[scale=0.8, transform shape]
        \node at (-0.5, 0.5){$\bz$:};
        \draw[general shadow={fill=gray}] (1,0) rectangle node {\Large $\bZ(\dd t)$} (2.5,1);
        \draw (2.5, 0)[pattern=my crosshatch dots] rectangle node{$\bz'$} (11.5, 1);
        \draw [-] (2.5, 0) -- (1, -1.5);
        \draw [-] (11.5, 0) -- (10, -1.5);
        \node at (-0.5, -2){$\bz'$:};
        \draw (1, -2.5)[pattern=my crosshatch dots] rectangle node{$\bz'$} (10, -1.5);
        \draw [->] (1,-3) node[left] {Time} -- (11.5,-3);
    \end{tikzpicture}
    \caption{An illustration of the proof of Lemma~\ref{lem:GD-redisual} that shows the first direction ($\GD(\bx) \leq -\sum_j w_j$). The first row represents the residual schedule before GD's processing. The second row represents the schedule we consider after GD's processing for $\dd t$ time units. GD processes the jobs $\bz(\tau)$, $\tau \in [0, \dd t)$, which is colored dark grey. $\bz'$ is a schedule of the remaining sizes after processing sizes $\bZ(\dd t)$. Since $\bz'$ is a feasible residual solution for remaining sizes $\bx' = \bx - \bZ(\dd t)$, we can upper bound $\GD(\bx)$ by comparing $\bz$ and $\bz'$. 
    Because $\bz'$ is a suffix of $\bz$, shifted by $\dd t$ time units, each job's completion time differs by exactly $\dd t$ units in both schedules. This time difference translates directly to the difference in the objective function, which is the total weight of (alive) jobs.  }  
\end{figure}
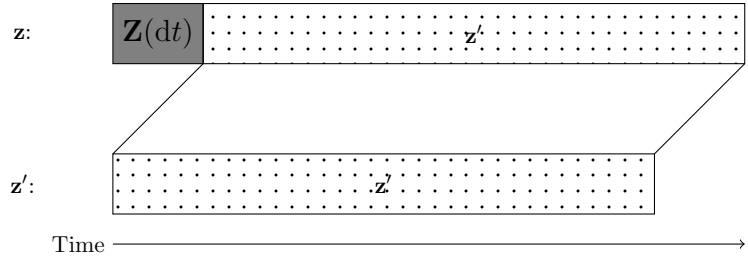

Recall that $\GD^\cP(\bx) := - \min_{\bz \in \cP} \nabla f(\bx) \cdot \bz$ denotes the change rate of the residual optimum for remaining sizes $\bx$ due to GD's processing. 
As mentioned earlier, the $\cP$ can be dropped for convenience, so $\GD(\bx)$ denotes the same concept. If GD uses speed s, we denote it as $\GD_s(\bx)$.
Note that GD with $s$ speed can process jobs at rates $s \bz$ if $\bz \in \cP$. 

\begin{claim}
    \label{claim:gd-scale}
    For any $s >0$, we have $\GD_s(\bx) = -  s\GD(\bx)$.
\end{claim}

The following lemma proves that the change in the optimal residual objective value due to GD is precisely equal to the total weight of the uncompleted jobs. The intuition is simple: the residual optimal schedule remains optimal for minimizing total weighted completion time until a new job arrives. Consequently, every job's completion time in the residual schedule decreases by one in each time step. This establishes the ``$\leq$" direction. The other direction can be proven similarly.

\begin{lemma}
    \label{lem:GD-redisual}
    When the residual optimum is $f(\bx) := \min_{\bz} \cI^\cP_\bx(\bz)$, we have $\GD(\bx) = - \sum_{x_{j}>0}w_{j}$.
\end{lemma}

\begin{proof}
    Let $\bz(\tau)$, $\tau \in [0, +\infty)$ be the optimal residual schedule for remaining size $\bx$, i.e., $\bz = \arg \min_{\bz'} \cI^{\cP}_{\bx}(\bz')$. Consider the decrease of the residual optimum in $\dd t$ units of time when we run GD. Let $\dd t$ be an infinitesimally small time step such that no job is completed during time $[0, \dd t)$. Recall that GD schedules $\bz(\tau)$, $\tau \in [0, \dd t)$. Let $\bZ(\tau') := \int_{\tau'' \geq 0}^{\tau'} \bz(\tau'') \dd \tau''$. Let $\bx' = \bx - \bZ(\dd t)$ be the remaining sizes after processing $\bz$ in $\dd t$ time. Note that $\bz'(\tau) := \bz(\tau + \dd t)$, $\tau \in [0, \infty)$ is a feasible schedule, for remaining job sizes $\bx'$. Then we have $f(\bx') \leq \cI_{\bx'}^\cP(\bz')$ because $f(\bx')$ minimizes $\cI^{\cP}_{\bx'}$ w.r.t.  remaining sizes $\bx'$, and $\bz'$ is a feasible schedule w.r.t. the same vector $\bx'$. Thus we have
    \begin{equation}
        \label{eqn:gd-residual-1}
        f(\bx) - f(\bx') \geq \cI^{\cP}_{\bx}(\bz)- \cI^{\cP}_{\bx'}(\bz') = \dd t \sum_{x_{j}>0}w_{j}
    \end{equation}
    This is because every job $j$ s.t. $x_j > 0$ has a $\dd t$ units larger completion time in $\bz$ than in schedule $\bz'$ and is still alive in the infinitesimally time step. Formally, $j$ completes in $\bz$ at the maximum (supremum) time $c$ such that $z_{cj} > 0$, and therefore it completes at time $c - \dd t$ in $\bz'$ since $\bz(c) = \bz'(c - \dd t)$.

To show the other direction, suppose GD schedules $\bz''(\tau)$ for $\tau \in [0, \dd t)$ and let $\bx'$ be the resulting remaining job sizes. Let $\bz'$ be the optimal residual schedule for $\bx'$; thus we have $f(\bx') = \cI^{\cP}_{\bx'}(\bz')$. Let $\bz$ be the schedule that first has $\bz''$ for $\dd t$ units of time and then exactly follows $\bz'$. Since $\bz$ is a valid residual schedule for remaining sizes $\bx$, we have $f(\bx) \leq \cI^{\cP}_{\bx}(\bz)$. Therefore, 
\begin{equation}
    f(\bx') - f(\bx) \geq \cI^{\cP}_{\bx'}(\bz') - \cI^{\cP}_{\bx}(\bz) = - \dd t \sum_{x_{j}>0}w_{j}
\end{equation}
As above, the equality holds as each job $j$'s completion time (such that $x_j > 0)$ differs exactly by $\dd t$ in the two schedules $\bz$ and $\bz'$.
This implies GD decreases the residual objective by at most $\dd t \sum_{x_{j}>0}w_{j}$. Thus we have $\GD(\bx) =  - \sum_{x_{j}>0}w_{j}$, as desired.
\end{proof}

The next lemma demonstrates that we can implement GD more efficiently by following the optimal residual schedule until a new job arrives. Upon a new job arrival, we recompute a new optimal residual schedule. The lemma is unnecessary if we are willing to compute the optimal residual schedule at every time step.

\begin{lemma}
    \label{lem:GD-redisual-2}
    Let $\bz$ denote the optimum residual schedule for remaining sizes $\bp(t)$. Suppose no jobs arrive during the time interval $[t, t']$. Then, we can assume wlog that GD processes jobs at rates $\bz(\tilde t)$ at each time $t + \tilde t$ for all $\tilde t \in [0, t' - t]$.
\end{lemma}
\begin{proof}
    The proof of this lemma is implicit in that of Lemma~\ref{lem:GD-redisual}. Eqn.~(\ref{eqn:gd-residual-1}) shows that processing along $\bz$, until a job is completed, is as effective as GD, which by definition is the most effective in decreasing the residual objective. Note that the remaining schedule must be an optimum residual schedule for the rest of the jobs. Then, we can recurse on the remaining sizes, until time $t'$ when a new job arrives and we  recompute the optimum residual schedule.
   
\end{proof}

The following lemma bounds the potential change when a job arrives or is completed. This is where we use supermodularity.

\begin{lemma}
    \label{lem:discrete-change-clairvoyant-int}
    If $\Phi(t)$ changes discontinuously due to a job's arrival or completion, the change is non-positive.
\end{lemma}
\begin{proof}
    If there is no arrival and completion, thanks to 
    Lemma~\ref{lem:GD-redisual-2}, it is straightforward to see $\Phi(t)$ is continuous in $t$ and differentiable.

    First, we observe that job completion in $A$'s schedule or the optimum solution's schedule has no effect on $f$ and therefore on $\Phi(t)$. This is because once a job $j$'s remaining size becomes 0, then the residual optimum schedule does not process it. Thus, we only need to consider job arrival.
    
    Suppose a new job $j$ arrives at a fixed time $t$. For brevity, we drop $t$ in the following. Let $\bone_{j^A}$ denote the unit vector that has 1 only the position corresponding to $j$ in $A$'s schedule. 
        Define $\bone_{j^O}$ analogously for the other copy of $j$ in $O$'s schedule. Here $j^A$ and $j^O$ refer to $j$'s copy in the schedule of $A$ and $O$, resp.  Dropping $t$ from the notation for brevity, $f(\bp^A \ || \  \bp^O)$ increases by: 
    \begin{align*}
        f(\bp^{A} \ || \ p_j \bone_{j^A} \ || \ \bp^O \ || \ p_j  \bone_{j^O}) - f(\bp^{A} \ || \ \bp^O)
    = \bar g_{\by}(A \cup j^A \cup O \cup j^O) - \bar g_\by(A \cup O)
    \end{align*}
where $\by = \bp^{A} \ || \ p_j  \bone_{j^A} \ || \ \bp^O \ || \ p_j  \bone_{j^O}$. 
Recall that the function $\bar g_\by(S)$ computes $f$ assuming that each  $j' \in S$ has size $\by \cdot \bone_{j'}$ and every job not in $S$ has zero size. Notably, we view $f$ as a set function $\bar g$ by fixing each job's (remaining) size. 
In the remaining proof, we use $g$ instead of $\bar g_\by$ for notational brevity. 
    
    By supermodularity of $g$ (due to the theorem condition), we have, 
    \begin{align*}
    & \ g(A \cup j^A \cup O \cup j^O) - g(A \cup O)  \\
    \geq \ & g(A \cup j^A \cup j^O) - g(A) \nonumber\\
    = \ &  g(A \cup j^A \cup j^O) - g(A \cup j^A ) + g(A \cup j^A) - g(A) \nonumber\\
    \geq \ & g(A  \cup j^O) - g(A) + g(A \cup j^A) - g(A) \nonumber\\
    = \ & 2 (g(A  \cup j^A) - g(A)),
    \end{align*}
where the last inequality follows from the fact that $j^O$ and $j^A$ both have the same size $p_j$ when $j$ arrives.

Now, we shift to bounding the change of $f(\bp^A)$, which increases by exactly
\begin{align*}
        f(\bp^{A} \ || \ p_j \bone_{j^A}) - f(\bp^{A})
    = g(A \cup j^A) - g(A)
\end{align*}

Therefore, $\Phi(t)$'s change is at most
$\frac{2}{\eps} \left( g(A \cup j^A) - g(A) - 
\frac{1}{2} \cdot 2(g(A \cup j^A) - g(A))\right) = 0
$, as desired.
\end{proof}

\begin{lemma}
    \label{lem:together-change}
    Consider any time $t$ when no job arrival or completion occurs. If $A$ and $O$ are given $s$-speed and $1$-speed respectively,  $\frac{\dd}{\dd t} f(\bp^A(t) \ || \ \bp^O(t))  \geq - (s + 1) ( W^A(t) + W^O(t) )$.
\end{lemma}
\begin{proof}
    We observe that processing jobs in $A$ and $O$ can be thought as a feasible schedule with speed $s+1$ ($A$ and $O$ are given speeds $s$ and $1$ respectively). How much can one decrease $f(\bp^A(t) \ || \ \bp^O(t))$ when given $s+1$-speed? Lemma~\ref{lem:GD-redisual}, together with Claim~\ref{claim:gd-scale}, implies that even GD, which is the most effective in decreasing the residual objective, can decrease at a rate up to the total weight of all jobs (of non-zero remaining sizes) times $s+1$ when given $s+1$ speed. Thus, the lemma follows.     
       
\end{proof}

\begin{corollary}
    \label{cor:running-clairvoyant-int}
    Consider any time $t$ when no job arrival or completion occurs. Then, $\frac{\dd}{\dd t} \Phi(t) \leq -W^A(t) + (1 + \frac{2}{\eps}) W^O(t)$ when GD is given $1+\eps$-speed.
\end{corollary}
\begin{proof}
    From Lemmas~\ref{lem:GD-redisual} and \ref{lem:together-change}, we have
    $\frac{\dd}{\dd t} \Phi(t)  = \frac{\dd}{\dd t}\frac{2}{\eps}\left(f(\bp^A(t)) -  \frac{1}{2}f(\bp^A(t) || \bp^O(t))\right)
    \leq \frac{2}{\eps} ( - sW^A(t) + \frac{s+1}{2} (W^A(t) + W^O(t)))  = -W^A(t) + (1 + \frac{2}{\eps}) W^O(t)$, where $s = 1+\eps$.
\end{proof}

By Lemma~\ref{lem:discrete-change-clairvoyant-int} and the fact that $\Phi(T) = \Phi(0) = 0$, we have $\int_{t=0}^{\infty} \frac{\dd}{\dd t} \Phi(t) \dd t \geq 0$. By Corollary~\ref{cor:running-clairvoyant-int}, we know that $\int_{t=0}^{\infty} W^A(t) \dd t \leq \frac{2 +  \eps}{\eps} \int_{t=0}^{\infty} W^O(t) \dd t$. Since the LHS and RHS are the total integral weighted flow time of GD and the optimum schedule respectively (Proposition~\ref{pro:integral-weight}), we have Theorem~\ref{thm:main-clairvoyant-int}.

\subsection{Proof of Theorem~\ref{thm:procesing-rate-view}}
\label{sec:proof-theorem-processing}

We continue to use the same notation $\bp^A(t)$ and $\bp^O(t)$ as we used in the previous section. We use the same meta-potential function. But, here we use fractional completion time for the residual optimum, so we have $f(\bx) := \min_\bz \cL^{\cP}_\bx(\bz)$ (see Eqn.~(\ref{eqn:frac-residual})).

$$ \Phi(t) = \frac{2}{\eps}\left(f(\bp^A(t)) - \frac{1}{2}f(\bp^A(t) || \bp^O(t))\right)$$

As before, each job $j$ has two distinct copies, which each appear in $\bp^A$ and $\bp^O$.

We first establish $f$'s supermodularity. 

\begin{lemma}
    \label{lem:second-thm-super}
    $f(\bx) := \min_\bz \cL^{\cP}_\bx(\bz)$ is (discrete-)supermodular.
\end{lemma}
\begin{proof}
Showing the lemma is equivalent to showing $- f(\bx)$ is submodular. Then, the objective in Eqn.~(\ref{eqn:frac-residual}) becomes 
$\max \sum_{j \in A(t)}   \int_{\tilde t \geq 0} - \frac{w_j}{p_j} \tilde t \cdot z_j(\tilde t) \dd \tilde t$. This is equivalent to $\max \sum_{j \in A(t)}   \int_{\tilde t \geq 0} (B - \frac{w_j}{p_j} \tilde t)  z_j(\tilde t) \dd \tilde t$ for any constant $B > 0$ due to the constraint $
\int_{\tilde t \geq 0} z_j(\tilde t) \dd \tilde t = x_j$
for all $j \in A(t)$; also adding a constant does not change a function's submodularity.  Further, if we let $B$ be a sufficiently large constant, we obtain the following equivalent formulation. 
\begin{align}
     \max \sum_{j \in A(t)} &  \int_{\tilde t \geq 0} \left(B -  \frac{w_j \tilde t}{p_j} \right)  z_j(\tilde t) \dd \tilde t & \label{eqn:frac-residual-flipped} \\
     \int_{\tilde t \geq 0} z_j(\tilde t) \dd \tilde t &\leq  x_j \quad  \forall j \in A(t) \nonumber \\
     \bz(\tilde t) &\in \cP \quad  \forall \tilde t \geq 0 \nonumber 
 \end{align}
Note that although the first constraint does not require equality, due to the big constant $B$, the optimum solution to the above LP must process $j$ to the full $x_j$ units. 
 
We now create a valuation function $v_{\tilde t}$ for each time $\tilde t$.
\begin{align*}
    v_{\tilde t}(\bz(\tilde t)) := \max_{0 \leq \by \leq \bz(\tilde t), \by \in \cP} \sum_{j \in A(t)} \left(B -  \frac{w_j \tilde t}{p_j} \right)  y_j
\end{align*}
In other words, here we view each time $\tilde t$ as an agent\footnote{Technically, achieving this result requires discretizing the time horizon into a finite set of sufficiently small unit times and then taking the limit. However, to maintain clarity in the presentation, we omit this technical detail as it is straightforward. } with valuation function $v_{\tilde t}$ and view each job $j$'s size/processing as a divisible resource, which exists by $x_j$ units. When agent $\tilde t$ receives processing units $\bz(\tilde t)$, it can process each job $j$ and it would like to maximize an objective linear in $\by$ subject to $\by \in \cP$ and $\by \leq z(\tilde t)$.

By the precondition of the theorem, $v_{\tilde t}$ is LS. Then max LP objective is equivalent to 
$$ \max \int v_{\tilde t} (z(\tilde t)) \dd \tilde t  \mbox{ s.t. } \int_{\tilde t \geq 0} z(\tilde t) \dd \tilde t \leq \bx$$
Since this is a convolution of LS valuations, thanks to Lemma~\ref{lem:LS-aggregation}, we have the max LP optimum is LS w.r.t. $\bx$. From Lemma~\ref{lem:LS-submodular}, we know that it is submodular. 
Thus, we have shown $f$ is supermodular. 
It is an easy exercise to show that this means $f$'s discrete-supermodularity (see Definition~\ref{def:d-super}) using the fact that $\bone_{A \cup B} \odot \bx = (\bone_{A} \odot \bx) \vee (\bone_{B} \odot \bx)$ and $\bone_{A \cap B} \odot \bx = (\bone_{A} \odot \bx) \wedge (\bone_{B} \odot \bx)$.

\end{proof}

The remaining analysis is almost identical to the proof of Theorem~\ref{thm:procesing-rate-view}. 
We start the analysis by considering $\Phi(t)$'s discontinuous changes. The proof of the following lemma is identical to that of Lemma~\ref{lem:discrete-change-clairvoyant-int}, which crucially used  $f$'s supermodularity.

\begin{lemma}
    \label{lem:discrete-change-processing-view}
    If $\Phi(t)$ changes discontinuously, the change is non-positive.
\end{lemma}
    
Let $\tilde W^A(t) := \sum_{j \in A(t)} \frac{w_j}{p_j}p^A_j(t)$ denote the total fractional weight of jobs at time $t$ in $A$'s schedule. Let $\tilde W^O(t)$ denote the corresponding quantity in $O$'s schedule. We will assume that GD follows the optimum residual schedule and  show that no algorithm can decrease the residual optimum  more effectively than GD.\footnote{In the proof of Theorem~\ref{thm:main-clairvoyant-int} we defined GD to be the most effective algorithm for decreasing the residual optimum, and we showed following the residual optimum schedule is as effective as GD. We make this small change, as we now know GD is really nothing but following the residual optimum.}
The following lemma quantifies the decrease rate of the (fractional) residual optimum. 

\begin{lemma}
    \label{lem:second-theorem-gd}
    For any time $t$ when no job arrival or completion occurs, we have $\frac{\dd}{\dd t} f(\bp^A(t)) = - s \tilde W^A(t)$ when we run GD with $s$ speed. Furthermore, for any algorithm with $s$ speed, we have $\frac{\dd}{\dd t} f(\bp^A(t)) \geq - s \tilde W^A(t)$.
\end{lemma}
\begin{proof}
    Fix a time $t$ and consider how much 
     $f(\bp^A(t))$ changes during $[t, t+ \dd t)$ when we run GD with $s$-speed. 
         Let $\bz = \arg \min_{\bz'} \cL_{\bp^A(t)}(\bz')$. Note that $j$ contributes to the objective $f(\bp^A(t))$ by $\int_{\tau \geq 0} \frac{w_j}{p_j} \tau z_{j}(\tau) \dd \tau$.
    GD schedules $\bz(\tau)$, $\tau \in [0, s \dd t)$ during time interval $[t, t+ \dd t)$ as it is given $s$ speed. Note that 
    $\bz(\tau)$, $\tau \in [s \dd t, \infty)$ is a feasible schedule (with 1-speed) for remaining job sizes $\bp^A(t) - \bZ(s \dd t)$ where $\bZ(\tau') := \int^{\tau'}_{\tau'' \geq 0} \bz(\tau'') \dd \tau''$. Therefore, $\bz'$, where  $\bz'(\tau) = \bz(s \dd t+ \tau)$, is a feasible solution to $\cL_{\bp^A(t) - \bZ(s \dd t)}$.

    Thus, $f(\bp^A(t))$ decreases by
    \begin{align*}
    &f(\bp^A(t)) - f(\bp^A(t)- \bZ(s \dd t)) \\
    \geq \ &f(\bp^A(t)) - \cL_{\bp^A(t) - \bZ(s \dd t)}(\bz') \\
    = \ &  \sum_j \left(\int_{\tau \geq 0} \frac{w_j}{p_j} \tau z_{j}(\tau) \dd \tau - \int_{\tau \geq s \dd t} \frac{w_j}{p_j} (\tau - s \dd t) z_{j}(\tau) \dd \tau \right) \\
    \geq \ & s \dd t \sum_j \frac{w_j}{p_j} \int_{\tau \geq s dt} z_j(\tau) \dd \tau = s \dd t (\tilde W^A(t) - \sum_j \frac{w_j}{p_j} Z_j(s \dd t)),
    \end{align*}
    which implies $\frac{\dd}{\dd t} f(\bp^A(t)) \leq - s \tilde W^A(t)$.

    Similarly, we can show the other direction. 
    Consider any feasible $s$-speed processing $\bZ'$ during the time interval; so, $j$ is processed by $Z'_j$ units. Let $\bz = \arg \min_{\bz'} \cL_{\bp^A(t) - \bZ'}(\bz')$ denote the optimum residual schedule for the remaining job sizes. 
    Let $\bZ' \circ \bz$ denote the 1-speed schedule where we first process $\bZ'$ for $s \dd t$ time units and subsequently schedule $\bz$. Note that $\bZ' \circ \bz$ is a feasible solution to $\cL_{\bp^A(t)}$. Therefore, $f(\bp^A(t))$ decreases by
    \begin{align*}
        &f(\bp^A(t)) - f(\bp^A(t)- \bZ') \\
    \leq \ &\cL_{\bp^A(t)}(\bZ' \circ \bz) - \cL_{\bp^A(t)- \bZ'}(\bz) \\        
     \leq \   &  \sum_j \left(\int_{\tau \geq 0} \frac{w_j}{p_j} (\tau + s \dd t) z_{j}(\tau) \dd \tau - \int_{\tau \geq 0} \frac{w_j}{p_j} \tau z_{j}(\tau) \dd \tau \right) + \sum_j \dd t \frac{w_j}{p_j} Z'_j\\  
    \leq \  & s \dd t   \sum_j \frac{w_j}{p_j} \int_{\tau \geq 0}  z_j(\tau) \dd \tau +  \dd t \sum_j  \frac{w_j}{p_j} Z'_j\\    
    =   \ & s \dd t \tilde W^A(t) +  \dd t \sum_j  \frac{w_j}{p_j} Z'_j, 
    \end{align*}
    
    which implies  $\frac{\dd}{\dd t} f(\bp^A(t)) \geq - s \tilde W^A(t)$, as desired. 
    
\end{proof}

This lemma allows us to  bound the continuous change of the second term of $\Phi(t)$. The proof is essentially identical to that of Lemma~\ref{lem:together-change}: it follows from the observation that we can view $A$ and $O$'s combined processing as a $s+1$-speed schedule and the residual objective can change at a rate up to the total fractional weight of jobs in the system, times $s+1$.

\begin{lemma}
    \label{lem:second-theorem-adv}
    Consider any time $t$ when no job arrival or completion occurs. If $A$ and $O$ are given $s$-speed and $1$-speed respectively,  $\frac{\dd}{\dd t} f(\bp^A(t) \ || \ \bp^O(t))  \geq - (s + 1) ( \tilde W^A(t) + \tilde W^O(t))$.
\end{lemma}

Combining the two lemmas we obtain the following corollary. 

\begin{corollary}
    \label{cor:running-clairvoyant}
    Consider any time $t$ when no job arrival or completion occurs. When GD is given $1+\eps$-speed, we have $\frac{d}{dt} \Phi(t) \leq -\tilde W^A(t) + \frac{2 +  \eps}{\eps} \tilde W^O(t)$.
\end{corollary}

By Lemma~\ref{lem:discrete-change-processing-view} and the fact that $\Phi(T) = \Phi(0) = 0$, it must be the case that $\int_{t = 0}^\infty \frac{\dd}{\dd t} \Phi(t) \dd t \geq 0$. By Corollary~\ref{cor:running-clairvoyant} we know that
$\int_{t = 0}^\infty \tilde W^A(t) \dd t \leq \frac{2 +  \eps}{\eps} \int_{t = 0}^{\infty} \tilde W^O(t) \dd t$ where the LHS and RHS are exactly the total weighted fractional flow time of GD and the optimum schedule (Proposition~\ref{pro:fractional-weight}). Thus, we have Theorem~\ref{thm:procesing-rate-view} for total fractional weighted flow time. Lemma~\ref{thm:conversion} allows us to obtain an algorithm that is $(1+\eps)$-speed $O(1 / \eps^2)$-competitive for the integer objective.



\subsection{Proof of Theorem~\ref{thm:resource-view-2}}

This section is devoted to proving Theorem~\ref{thm:resource-view-2}. 
The high-level idea is to show that when we take the feasible processing rate view, we satisfy the condition of Theorem~\ref{thm:procesing-rate-view}, i.e.,  
$v(\bx) := \max_{0 \leq \bz \leq \bx, \bz \in \cP} \bc \cdot \bz$ is  linear substitute. Note that $\bz \in \cP$ if and only if $\bz_j = u_j(\by_j)$ for all $j$ for some non-negative vectors $\{\by_j\}_{j \in J}$ such that $\sum_j \by_j \leq \bone$. 

For any price vector $\bq$, let $\bx(\bq)$ denote the unique solution to: 
\begin{align*}
    \max_{\bx'} \sum_{i \in J} &(c_i - q_i) x'_i \\
    \bx' &\in \cP    
\end{align*}
To see that this is well defined, take the resource view of the CP. Recall in the resource view, each job $i$ is associated with a valuation function $u_i$ and the divisible resources $[0, 1]^d$ are allocated to the jobs. If $\by_i$ denotes the resources job $i$ receives, we have the following equivalent CP.
\begin{align}
        \max \sum_{i \in J} &(c_i - q_i) u_i(\by_i) \label{equ:wal-equ}\\
    \sum_{i \in J} \by_i &\leq \bone \nonumber \\
                    \by_i &\geq \textbf{0}  \quad \forall i \in J \nonumber
\end{align}

It is straightforward to see that this CP has a unique solution due to the theorem condition that $u_i$ is strictly concave for all jobs $i$.

Without loss of generality, we assume $\bq < \bc$. This is because jobs $i$ with $c_i - q_i \leq 0$ do not play any roles and therefore we can pretend that they do not exist. Now consider the Walrasian market where each job/agent $i$ has a valuation function $(c_i - q_i) u_i$ for divisible resources $[0, 1]^D$. Because $u_i$ is LS for all $i$, which is a precondition of Theorem~\ref{thm:resource-view-2}, $(c_i - q_i) u_i$ is LS for all $i$ (LS is closed under a positive scalar multiplication). 
Thus, the CP formulates the social surplus maximization problem where all agents have LS valuations. In the following $\bb$ and $\bb'$ are uniquely defined as $u_i$ is assumed to be strictly monotone and strictly concave.

\begin{lemma}
    \label{lem:mono-market}
    Consider any $\bq \leq \bq' \in \R^{|J|}$. 
    Let $\bb$ be the equilibrium price vector for the market with valuations $(c_i - q_i) u_i$, $i \in J$ along with divisible resources $[0, 1]^d$. Similarly, let $\bb'$ be the  equilibrium price vector for $(c_i - q'_i) u_i$, $i \in J$. We have $\bb' \leq \bb$.
\end{lemma}
\begin{proof}
    To prove this lemma, we use the tatonnement process described in \cite{milgrom2009substitute}.\footnote{Although the process is described for indivisible resources that exist in multiple units, the same proof extends to divisible resources. The proof can be much simplified if we assume valuation functions are strictly concave, as shown in \cite{ausubel2004auctioning}.}
    The tatonnement process is governed by a price vector $\hat \bb(\tau)$; the vector monotonically changes as $\tau$ increases.     We say that all resources are under-demanded\footnote{The over-demanded case is symmetric.} under the current price vector $\bb''$ if there exists an allocation $\by''$ that maximizes everyone's (here, every job $i$'s) quasi-linear utility and $\sum_i \by''_i \leq \bone$. Here, if job $i$'s utility is $(c_i - q_i) u_i$, then $\by''_i \in \+D((c_i - q_i) u_i, \bb'')$.
    For brevity, we will equivalently say that $\bb''$ is an under-demand price vector. 
    It is shown in \cite{milgrom2009substitute} that starting with any under-demand price vector $\hat \bb(0)$, we can monotonically decrease $\hat \bb(\tau)$ to an under-demand equilibrium price, which yields the maximum social surplus.
    
    Let $\by$ be the unique allocation achieving the unique equilibrium with the price vector $\bb$.
    Define $\by'$ analogously for $(c_i - q_i')u_i$ with the price vector $\bb$.
   
    Since for all $i$, $c_i - q_i \geq c_i - q_i'$ and $\by_i' \in \+D((c_i - q_i') u_i, \bb) = \+D( (c_i - q_i) u_i, \frac{ c_i - q_i}{ c_i - q_i'}\bb )$, we have  $\by_i' \leq \by_i$ due to the first property of Theorem~\ref{thm:resource-view-2}. Thus, $\sum_i \by'_i \leq \bone$, meaning $\bb$ is an under-demand price vector when each $i$ has valuation $(c_i - q'_i)u_i$. Thanks to the above tatonnement process, we conclude $\bb' \leq \bb$. 
\end{proof}

Now we are ready to prove that $v(\bx)$ is LS. 
Say $q'_j > q_j$ and $q'_i = q_i$ for all $i \neq j$, so we only increased job $j$'s price. 
Consider any $i \neq j$. We use $\by$ and $\by'$ as we defined in the proof of Lemma~\ref{lem:mono-market}. Since $u_i$ is strictly concave, we have, 
\begin{align*}
\{\by_i\} &=  \+D((c_i - q_i)u_i, \bb) = \+D(u_i, \bb / (c_i - q_i)) \\
\{\by'_i\} &=  \+D((c_i - q_i') u_i, \bb') = \+D((c_i - q_i) u_i, \bb') = \+D(u_i, \bb' / (c_i - q_i))
\end{align*}

From Lemma~\ref{lem:mono-market}, we know  $\bb' \leq \bb$. Thus, due to the second property stated in Theorem~\ref{thm:resource-view-2}, we have $u_i(\by_i') \geq u_i(\by_i)$. 

Knowing that $x_i(\bq) = u_i(\by_i)$ and $x_i(\bq') = u_i(\by_i')$, we conclude 
\begin{equation}
    \label{eqn:qq}
x_i(\bq') \geq x_i(\bq)
\end{equation}

\begin{lemma}
    \label{lem:demand-unique}
    If $\bq > \textbf{0}$, then the demand set $\mathcal{D}(v, \bq)$ is a singleton set. Further, we have $\mathcal{D}(v, \bq) = \{ \bx(\bq)\}$. 
\end{lemma}
\begin{proof}
The quasi-linear utility, $v(\bx'
) - \bq \cdot \bx'$ is maximized only when $\bx' \in \cP$ since $\bq 
 > \textbf{0}$. Then, the quasi-linear utility maximization can be written as $\max_{\bx'} \bc \cdot \bx' - \bq \cdot \bx'$ s.t. $\bx' \in \cP$ which has a unique solution as discussed above. The solution is $\bx(\bq)$ by definition. 
\end{proof}

Thus, if $\bq > \textbf{0}$, then we have $\bq' > \textbf{0}$ as well. Due to Lemma~\ref{lem:demand-unique} and Eqn.~(\ref{eqn:qq}), we have shown that $v$ is LS. Therefore, we have Theorem~\ref{thm:resource-view-2} from Theorem~\ref{thm:procesing-rate-view}.

We extend the proof to an arbitrary $\bq \geq 0$.
Let $J^0$ denote the set of jobs $j$ with $q_j = 0$, and $J^+$ the other jobs. It is an easy exercise to show that 
\begin{equation}
    \label{eqn:char}
\+D(v, \bq) = \{ \bx(\bq) + \sum_{j \in J^0} \lambda_j \bone_j  \; | \; \lambda_j \geq 0 \; \forall j \in J^0\}.
\end{equation}
In other words, the quasi-linear utility is maximized by purchasing any amount of free jobs on top of $\bx(\bq)$. Consider any $\bx \in \+D(v, \bq)$ and increase job $j$'s price, so we have $\bq'$ such that $q'_j > q_j$ and $q'_i = q_i$ for all $i \neq j$. For any $i \in J^0$ such that $i \neq j$, it is trivial to see that we can find $\bx'  \in \+D(v, \bq')$ such that $x'_i \geq x_i$ as job $i$ is free. If $i \in J^+$, we have $x_i(\bq') \geq x_i(\bq) = x_i$, where the inequality follows from Eqn.~(\ref{eqn:qq}) and the equality follows from Eqn.~(\ref{eqn:char}) with the facts that  $\bx, \bx(\bq) \in \+D (v, \bq)$  and $i \in J^+$.

\section{Applications of the Theorems}
    \label{sec:theorem-applications}

In this section, we discuss some problems captured by our theorems. 
See Section~\ref{sec:applications} for their definition.

\subsection{Matroid Scheduling}

We can formulate the matroid scheduling problem as an instance of the PSP as follows: 
$$ \cP = \{ \bz \ | \ \sum_{j \in S} z_j \leq r(S) \quad \forall S \subseteq J \}$$
where $r$ is the rank function of the given matroid $\+M = (J, \+I)$ and $z_{j}$ is the processing rate of job $j$. For a cost vector $\bc$, $\bar v(\bone_S) = \max_{\bz \leq \bone_S, \bz \in 2^J, \bz \in \cP} \bc \cdot \bz$ is equivalent to the weighted rank function of a matroid, which is known to be GS (Section~\ref{sec:gs}).
Note that $\bar v(\bone_S)  = v(\bone_S)$  for all $S \subseteq J$ because the polytope $\cP$ is integral, where  $v(\bx) = \max_{\bz \leq \bx, \bz \in \cP} \bc \cdot \bz$ (see Theorem~\ref{thm:procesing-rate-view}).

To show $v$ is LS, we prove that the concave closure of $\bar v$ is equivalent to $v$. For the concave closure definition, see Theorem~\ref{thm:gs-concave-closure}. Note that here we used $\bar v$ to distinguish the discrete valuation from the continuous valuation $v$. 

\begin{claim}
    For $\bx \in [0,1]^J$, $v(\bx) = \bar v^+(\bx)$.
\end{claim}
\begin{proof}
    First, we show the ``$\leq$" direction. Let $\bz \in \cP$ such that $v(\bz) = v(\bx)$. 
    Since $\cP$ is integral, we can express $\bz$ as a convex combination of the subsets of $J$, so $\bz = \sum_{S \in \mathcal{I}} \lambda_S \cdot \bone_S$, where $\sum_{S  \in \mathcal{I}}\lambda_S = 1$. Then, $v(\bx) = \bc \cdot \bz =  \sum_{S \in \mathcal I} \lambda_{S} (\bc \cdot  \bone_S) = \sum_{S \in \mathcal I}\lambda_S ~ \bar v(\bone_S) \leq v^+(\bx)$.

    To show the other direction, let $\bar v^+(\bx) = \sum_S \bar v(\bone_S) \lambda_S$ for some   $\{\lambda_S\}_{S \subseteq J}$ with $\sum_S \lambda_S = 1$ and $\sum_S \lambda_S \bone_S = \bx$. Let $\pi(S)$ be $S' \subseteq S$ such that $\bone_{S'} \in  \cP$ and 
    $\bar v(\bone_{S'}) = \bar v(\bone_{S})$. We then have
    \begin{equation}
        \label{eqn:matroid-11}
        \bar v^+(\bx) = \sum_S \bar v(\bone_{\pi(S)}) \lambda_S = \sum_S   \lambda_S ~ \bc \cdot \bone_{\pi(S)}
    \end{equation}

    Let $\bx' := \sum_S \lambda_{S} \bone_{\pi(S)}$. Note that $\bx' \leq \bx$. Further, $\bx' \in \cP$ due to $\cP$'s integrality. By definition of $v$, we know $v$ is monotone. Thus, we have, 
    \begin{equation}
        \label{eqn:matroid-12}
        v(\bx) \geq v(\bx') = \bc \cdot \bx' = \sum_S \lambda_S ~ \bc \cdot \bone_{\pi(S)}
    \end{equation}
    From Eqn.~(\ref{eqn:matroid-11}) and (\ref{eqn:matroid-12}), we conclude $v(\bx) \geq \bar v^+(\bx)$, as desired. 
   
\end{proof}

Since $\bar v^+$ is LS, so is $v$.  Therefore, Theorem~\ref{thm:procesing-rate-view} establishes GD as a scalable algorithm. 
Finally, if  $\{ \lambda_S \}_{S \in \mathcal{I}}$ is the convex combination of independent sets that achieves $v^+(\bx)$, we can implement our algorithm by preemptively scheduling each independent set $S$ in proportional to $\lambda_S$.

\subsection{Generalized Flow Scheduling}

Recall that we are given a network $G=(V,E)$, where $J \subseteq V$ is the set of source nodes. Arc $e$ has capacity $u_e$ and flow gain factor $\gamma_e$; a flow of value $f$ becomes $\gamma_e f$ after flowing through $e$. 
Job $j$ gets processed at a rate equal to its (outgoing) net flow. Formally, node $v$ has excess capacity $b_v$. Sources have non-negative excess capacities. A job $j$ gets processed at a rate of $z_j$, when its net flow is $z_j$ for  $\bz \in \cP$, which is defined as follows.
\begin{equation} 
    \label{eqn:p-gf}
\cP = \set{\bz  = \langle \sum_{w} (f_{jw} - \gamma_{wj} f_{wj})  \rangle_{j  \in J} \; \Big| \;  \sum_{w} (f_{vw} - \gamma_{wv} f_{wv})\leq b_v \; \forall v \in V; f_e \leq u_{e} \; \forall e \in E; \vecf \geq 0}
\end{equation}

Theorem~\ref{thm:gf-ls} shows $v(\bx) = \max_{\bz \in \cP} \bc \cdot \bz$ is LS. Note that our network can ensure $\bz \leq \bx$ by setting the excess capacity of the node that represents job $j$ to be $x_j$. Thus, we immediately have a scalable algorithm due to Theorem~\ref{thm:procesing-rate-view}.

\subsection{Speed-up Curves on Partitioned Resources}

We consider this problem from the Resource View. Then we need to show that the utility $u_j$ is linear-substitute and satisfies two desired properties in Theorem~\ref{thm:resource-view-2}. Without loss of generality, we can assume the utility is determined by only one univariate function that is $u_j(\bx) = g(\sum_{d\in [D]} a_d x_d)$. We drop $j$ for brevity henceforth.

We begin by showing a simple greedy algorithm being optimal. The proof follows from a simple swapping argument. 
\begin{lemma}
    \label{lem:concave-greedy-OPT}
    For any price vector $\bq$ on resources, the greedy algorithm selects elements fractionally in decreasing order of $a_d / q_d$ until no marginal increment, finds a solution maximizing the quasi-linear utility, i.e., a solution in  $\+D(u, \bq)$.
\end{lemma}
\begin{proof}
    Fix some price vector $\bq$ on resources.  Assume that $a_1 / q_1 \geq a_2 / q_2 \geq ...\geq a_D / q_D$. Suppose we have an optimum solution $\bx$ where  $x_i < 1$ and $x_j > 0$ for some $i < j$. Then, we maximally increase $x_i$ by $a_j \delta / a_i$ and decrease $x_j$ by $\delta$ for some $\delta 
 >0 $
     until $x_i$ becomes 1 or $x_j$ becomes 0. This    preserves the value of $\sum_d a_d x_d$ without increasing $\bq \cdot \bx$. By repeating  this, we have the lemma. 
\end{proof}

First, we prove that $u$ satisfies linear-substitute. Although the claim is known to be true (see Section 
 5 in \cite{murota2022discrete}), we include it as it is a nice warm-up before proving the subsequent claims. 

\begin{claim}  
    $u$ is linear-substitute.
\end{claim}
\begin{proof}
    Consider the greedy algorithm in Lemma~\ref{lem:concave-greedy-OPT}. Now we increase the price of resource $d$. In the new greedy order, the resource $d$ can only be considered later; other elements keep the same relative order. Thus $\bx^*_i$ can not decrease for all $i \neq d$ because they are processed earlier than the original greedy order. Therefore, we have the claim.
\end{proof}

The following claim shows that when we increase the price vector $\bq$ simultaneously for each coordinate at the same scale, $u$ can only demand less.

\begin{claim}
    For any $\by \in \+D(u, \bq)$ and $\lambda \geq 1$, there exists $\by' \in \+D(u, \lambda \bq)$ such that $\by \geq \by'$.
\end{claim}
\begin{proof}
    Consider the greedy algorithm in Lemma~\ref{lem:concave-greedy-OPT}. Note that increasing the price vector $\bq$ to $\lambda \bp$ does not change the order of $a_d / q_d$. Thus the stopping point in the execution of the greedy algorithm comes earlier as the marginal increment becomes zero with fewer resources in the order, which gives the proof.
\end{proof}

Further, we show that when we increase the price vector $\bq$, the utility would not increase.

\begin{claim}
    For any $\by \in \+D(u, \bq)$ and $\bq' \geq \bq$, there exists $\by' \in \+D(u, \bq')$ such that $u(\by) \geq u(\by')$.
\end{claim}
\begin{proof}
    For the sake of contradiction, suppose $u(\by) < u (\by')$.  
Let $\by'$ be the solution constructed by the greedy algorithm described in the proof of Lemma~\ref{lem:concave-greedy-OPT} with price vector $\bq'$. Define $\by$ analogously with $\bq$. Because of the assumption $u(\by) < u (\by')$ and $g$'s monotonicity, it must be the case that 
$\sum_i a_i y_i < \sum_i a_i y_i'$. Let $\by''$ be the ``prefix'' of $\by'$ such that $z :=\sum_i a_i y_i = \sum_i a_i y_i''$. In other words, we run the greedy with $\bq'$ but stop when  $z  = \sum_i a_i y_i''$. 
Let $k$ be the smallest such that $y_k' > y_k''$. The greedy with respect to $\bq'$ further increases $\by''$ since the marginal gain $g(z + \eps a_k) - g(z) - \eps \bq'_k > 0$ for a sufficiently small $\eps >0$. Therefore, if $\by_k < 1$, the greedy must have continued increasing $\by$ since $g(z + \eps a_k) - g(z) - \eps q_k \geq g(z + \eps a_k) - g(z) - \eps \bq'_k > 0$. Since the greedy with respect to $\bq$ considers resources in the order of $1, 2, \ldots$, it means $\by_1 = \by_2 = \ldots = \by_k = 1$; yet $\by''_k < 1$ and $\by''_{k+1} = \by''_{k+2} = \ldots = \by''_D  = 0$. This is a contradiction to
$\sum_i a_i y_i < \sum_i a_i y_i'$.
\end{proof}
Finally, we discuss why it is wlog to assume that $u$ is strictly monotone and strictly concave. To obtain a strictly concave function that is arbitrarily close to $u$, we replace each $x_d$ with $x_d^{1 - \delta}$ for an arbitrary small $\delta > 0$. Further, we replace $g$ with a strictly monotone function that is arbitrarily close to $g$. It is an easy exercise to show this new function is strictly concave and arbitrarily close to the original function $u$ on $[0, 1]^D$. We then pretend an infinitesimal unit $\eps$ of each resource as a distinct item and order these infinitesimal-sized items in decreasing order of 
$a_d (( k\eps)^{1 - \delta} - ((k-1)\eps)^{ 1 - \delta}) / (p_d \eps)$, which corresponds to the effectiveness of  the $k$-th $\eps$-sized copy of resource $d$. It is straightforward to check that the whole analysis remains unchanged. Further, if $u$ does not depend on some resources, then for each of such resources $d$ we can add $\lambda_d x_d^{1-\delta}$ for a sufficiently small $\lambda_d > 0$, so this additional term essentially has no effect on $u$ while ensuring $u$ being strictly monotone. 

Thus, this application admits 
 a scalable algorithm due to Theorem~\ref{thm:resource-view-2}.

\subsection{Unrelated Machines for Unweighted Jobs}
    \label{sec:unrelated}

In this section we show how to obtain a $(1+\eps)$-speed $O(1 / \eps)$-competitive algorithm for total unweighted flow time using gradient descent; the weighted case is presented in the following section. Our algorithm is migratory, unlike the algorithms in \cite{ChadhaGKM09,AnandGK12}. (We also reproduce the non-migratory result in \cite{AnandGK12} in  Section~\ref{sec:unrelated-immediate}, focusing on supermodularity and gradient descent.) 

Unfortunately, we can not directly use Theorem~\ref{thm:main-clairvoyant-int} for the following reason: the problem we consider allows migration of jobs across machines. The residual optimum of a migratory schedule does not seem to be LS. Thus, we use 
the non-migratory residual optimum, which turns out to be GS. Interestingly, our resulting gradient descent algorithm is migratory. Nevertheless, the analysis will be very similar to the proof of Theorem~\ref{thm:main-clairvoyant-int}. Thus, we will adopt the same notation throughout this section.

We first describe the residual optimum for this problem. Let $f(\bx)$ denote (the cost of) the min-cost non-migratory schedule of jobs, where each job $j$ has (remaining) size $x_j$. It is well known that it can be computed by min-cost bipartite matching. For completeness, we describe it here: Consider a bipartite graph $G = (J, M'; E)$ where there is a unique node,  indexed by $(i, k)$, for each pair of machine $i$ and $k \in \Z_{> 0}$. There is an edge between $j \in J$ and $(i, k) \in M'$. Choosing the edge means scheduling  job $j$ on machine $i$ in the $k$-th \emph{reverse} order non-preemptively, thus the edge has cost $k \cdot p_{ij}$, where $p_{ij} = x_{j} / \lambda_{ij}$. Intuitively, this means there are $k$ jobs behind job $j$ on the machine, including itself---thus $j$ delays $k$ jobs---when following SRPT on each machine. The residual optimum $f(\bx)$ is the min cost of a bipartite matching for this instance.

 \begin{lemma}
    \label{lem:unrelated-super}
    $\bar g_\bx(J') := f(\bx \odot \bone_{J'})$ is  (discrete)-supermodular.
 \end{lemma}
\begin{proof}
    From Section~\ref{sec:gs-examples}, we know that maximum weight bipartite matching is GS. 
    Therefore, $ - \bar g_\bp$ is GS and submodular. 
\end{proof}

As before, we use $\GD(\bx)$ to denote the max decrease rate of the residual optimum. We will see that following the optimum residual schedule until a new job arrives is as effective as the true GD; that is, it decreases the residual optimum at a rate equal to $\GD(\bx)$.

Since the analysis is very similar to  the proof of Theorem~\ref{thm:main-clairvoyant-int}, we only highlight the differences. We use the same potential function. Lemma~\ref{lem:unrelated-super} gives supermodularity of the residual optimum. Thus, we can show all discrete changes of $\Phi$ is non-positive following the same argument in the proof of Lemma~\ref{lem:discrete-change-clairvoyant-int}.

\begin{lemma}
    \label{lem:GD-redisual-unrelated}
    When the residual optimum is $f(\bx)$, we have  $\GD(\bx) = -\sum_{x_{j}>0} 1$. 
\end{lemma}
\begin{proof}
    Proving $\GD(\bx) \leq -\Sigma_{x_{j}>0} 1$ is identical to the first part proof of Lemma~\ref{lem:GD-redisual}. In fact, we can decrease $f(\bx)$ at a rate of at least $\Sigma_{x_{j}>0} 1$ by following the residual optimum schedule. 
 
Showing the other direction needs special care because while the residual optimum schedule is non-migratory, changing the job assignment could help decrease the residual optimum. Assume wlog that $\bx > 0$, because we can ignore jobs $j$ such that $x_j = 0$.
Suppose GD processes $j(i)$ on machine $i$ for a sufficiently small unit time $\dd t$, so no job completes during the time slot. Let $\bx'$ be the resulting remaining job sizes and $\sigma'$ be the optimal non-migratory schedule for $\bx'$, where $j$ is scheduled on machine $\psi(j)$. Let $\sigma$ be the non-migratory schedule obtained from $\sigma'$ by increasing each job's remaining size from $x'_j$
to $x_j$, without changing the assignment $\psi$ and the job ordering on each machine. So, $\sigma$ is a feasible non-migratory schedule for $\bx$, but not necessarily optimum. Let $\cost(\sigma)$ denote the total completion time of schedule $\sigma$. Clearly we have,
$$
    \GD(\bx) \cdot \dd t \geq \cost(\sigma')  - \cost(\sigma).
$$
Thus, it suffices to lower bound $\cost(\sigma')  - \cost(\sigma)$. For notational brevity, let $x'_{ij} := x'_j / \lambda_{ij}$, which is $j$'s processing time on $i$ when its remaining size is $x'_j$. Similarly, let $x_{ij} := x_j / \lambda_{ij}$. Define $P_{i, <  k} := \sum_{j': x'_{ij'} < x'_{ik}} x'_{ij'}$ to be the total remaining processing time of jobs smaller than job $k$ on machine $i$ w.r.t. $\bx'$ in $\sigma'$, breaking ties in an arbitrary but fixed order. Let $N_{i, \geq k} := \sum_{j': x'_{ij'} \geq x'_{ik}} 1 $ denote the number of jobs of  remaining sizes no smaller than job $k$ on machine $i$ w.r.t. $\bx'$ in $\sigma'$.

Since $j(i)$ gets processed on machine $i$, we have $x_{j(i)} - x'_{j(i)} = \lambda_{i, j(i)} \dd t$. The job $j(i)$ is on machine $\psi(j(i))$, and the number of jobs behind it\footnote{Note that on each machine $i$ jobs are ordered following the SRPT order, i.e. in decreasing order of $x_{ij}$.}, including itself, is exactly $N_{\psi(j(i)), \geq j(i)}$. Thus, due to  $x'_{j(i)}$ increasing to $x_{j(i)}$, $j(i)$'s contribution to $\cost(\sigma') - \cost(\sigma)$ is exactly,
\begin{equation}
    \label{eqn:nec3}
    N_{\psi(j(i)), \geq j(i)} x'_{\psi(j(i)),j(i)} - N_{\psi(j(i)), \geq j(i)} x_{\psi(j(i)), j(i)}
    = - N_{\psi(j(i)), \geq j(i)} \frac{\lambda_{i, j(i)}}{\lambda_{\psi(j(i)), j(i)}} \dd t
\end{equation}

Therefore, we have
$$
    \GD(\bx) \geq - \sum_i N_{\psi(j(i)), \geq j(i)} \frac{\lambda_{i, j(i)}}{\lambda_{\psi(j(i)), j(i)}} \dd t
$$

We complete the proof by showing 
\begin{equation}
    N_i \geq N_{\psi(j(i)), \geq j(i)} \frac{\lambda_{i, j(i)}}{\lambda_{\psi(j(i)), j(i)}}
\end{equation}
for all $i$, 
where $N_i := |\psi^{-1}(i)|$ is the number of jobs assigned to machine $i$. To see this, observe the following: 
\begin{equation}
    \label{eqn:nec}
P_{i, <  j(i)} + N_{i, \geq j(i)} \frac{x'_{j(i)}}{ \lambda_{i, j(i)}} \geq
P_{\psi(j(i)), <  j(i)}  + N_{\psi(j(i)), \geq j(i)} \frac{x'_{j(i)}}{\lambda_{\psi(j(i)), j(i)}}, 
\end{equation}
This is because the LHS and RHS measure the increase of the objective due to $j(i)$'s presence on machines $i$ and  $\psi(j(i))$, respectively (in other words, the LHS is the increase of the residual objective on machine $i$ 
when we move $j(i)$ from $\psi(j(i))$ to $i$, and the RHS is the decrease of the residual objective on machine $\psi(j(i))$ for that move): The first term is $j(i)$'s start time and the second is the number of jobs delayed by $j(i)$, including itself, multiplied by $j(i)$'s remaining size on the machine. Thus, the inequality is a necessary condition for $\sigma'$ being the optimal residual schedule for $\bx'$.

Multiplying both sides in Eqn.~(\ref{eqn:nec}) by $\frac{\lambda_{i, j(i)}}{x'_{j(i)}} = \frac{1}{x'_{i,j(i)}}$, we have
\begin{equation}
    \label{eqn:nec2}
\frac{P_{i, <  j(i)}}{x'_{i, j(i)}} + N_{i, \geq j(i)}  \geq
\frac{P_{\psi(j(i)), <  j(i)}}{x'_j /\lambda_{i, j(i)}}  + N_{\psi(j(i)), \geq j(i)} \frac{\lambda_{i, j(i)}}{\lambda_{\psi(j(i)), j(i)}}. 
\end{equation}

For the first term, we have $\frac{P_{i, <  j(i)}}{x'_{i, j(i)}} = \sum_{j': x'_{ij'} < x'_{i, j(i)}} x'_{ij'} / x'_{i, j(i)} \leq \sum_{j': x'_{ij'} < x'_{i, j(i)}} 1$. Thus, it is at most the number of jobs that have smaller processing times than $j(i)$ on machine $i$. Thus, the LHS is at most the number of jobs assigned to machine $i$, i.e., $N_i$, as desired. 
\end{proof}

As mentioned in the above proof, following the non-migratory residual optimum schedule decreases the residual as much as GD. Further, it is an easy exercise to show that we do not have to compute a residual optimum until a new job arrives. All the remaining analysis remains the same and is omitted.

\subsection{Unrelated Machines for Weighted Jobs}
    \label{sec:unrelated-weighted}

In this section, we extend the GD algorithm to handle weighted jobs and give an algorithm that is $(1+\eps)$-speed $O(1 / \eps^3)$-competitive algorithm for total weighted flow time. One may try the previous residual LP (Eqn.~(\ref{eqn:frac-residual})) where 
    $$\cP = \Big\{ z_j(\tilde t) = \sum_{i} \lambda_{ij} y_{ij}(\tilde t) \; | \; \sum_{i} y_{ij}(\tilde t) \leq 1 \ \forall j, \tilde t ; \; \sum_j y_{ij}(\tilde t) \leq 1 \ \forall i, \tilde t \Big\} $$
Unfortunately, this objective does not seem to be LS. 
Instead, we consider essentially the same LP that was used in \cite{AnandGK12} for their dual fitting analysis. 
However, our algorithm is very different from the immediate-dispatch algorithm that was given in  \cite{AnandGK12}. 
For notional simplicity, we may use $p_{ij} := p_j / \lambda_{ij}$ to denote the processing time of job $j$ on machine $i$ when $j$ is processed on machine $i$ only. 
\begin{align}
    \mbox{Residual LP $\cL_{\bx}:$} \quad \quad f(\bx) := \min \underbrace{\sum_{i,j} \frac{w_j \lambda_{ij}}{p_{j}} \int_{\tilde t \geq 0} \tilde t \cdot z_{ij}(\tilde t) \dd \tilde t}_{(*)} &+ \underbrace{\sum_{i,j} w_j \int_{\tilde t \geq 0} z_{ij}(\tilde t) \dd \tilde t}_{(**)} \label{eqn:lp-um-l1} \\
    \sum_{j} z_{ij}(\tilde t) &\leq 1 \quad \forall i, \tilde t \nonumber \\
    \sum_{i} \int_{\tilde t \geq 0} \lambda_{ij} z_{ij}(\tilde t)\dd \tilde t &= x_j \quad \forall j \nonumber
\end{align}
This new objective will be shown to be supermodular from the LS property (see Lemma~\ref{lem:apx-ur-sup}). Notice that this LP does not force a job to be processed on at most one machine; thus $\bz(\tilde t)$ may not be a feasible schedule. For this reason, we can not directly apply the GD algorithm in Section~\ref{sec:Substitutes-Empower-GD}, which schedules $\bz(\tilde t = 0)$. 

 To handle this issue, we consider the following \emph{approximate} GD: For each machine $i$, we only process one unit of jobs that appear the earliest in $\bz(\tilde t)$, ensuring that each job $j$ gets processed up to the fraction of the job assigned to the machine. Formally, 
 let $\tilde t_i$ be the earliest time that $\sum_j \int_{\tilde t=0}^{\tilde t_i} \frac{z_{ij}(\tilde t)}{p_{ij}} \dd \tilde t = 1$. Then we process $\hat z_{ij} := \int_{\tilde t=0}^{\tilde t_i} \frac{z_{ij}(\tilde t)}{p_{ij}} \dd \tilde t$ fraction of job $j$ on machine $i$.
We continue this schedule until the fraction of a job on some machine is completed. For brevity, we may call this approximate GD as GD. 

\begin{remark}
    If we use the true GD algorithm, which decreases the residual optimum objective the most, we can still show the same guarantee. However, for the sake of analysis, we will consider the approximate GD. 
\end{remark}

We first verify that  $\hat \bz$ is a feasible schedule. By definition of the algorithm, we have $\sum_{j} \hat z_{ij} \leq 1$ for all $i$. This, together with the following claim, shows that $\hat z_{ij}$ is a fractional matching between jobs and machines. Since a fractional bipartite matching can be expressed as a convex combination of integral matchings, we have that $\hat z_{ij}$ is indeed a feasible schedule. 
\begin{claim}
    For any job $j$, we have $\sum_{i} \hat z_{ij} \leq 1$.
\end{claim}
\begin{proof}
    From the second constraint of LP, we have\\
        $$\sum_{i} \hat z_{ij} = \sum_i \int_{\tilde t=0}^{\tilde t_i} \frac{z_{ij}(\tilde t)}{p_{ij}} \dd \tilde t \leq \sum_i \int_{\tilde t\geq 0}^{\tilde t_i} \frac{\lambda_{ij}z_{ij}(\tilde t)}{p_j}  \dd \tilde t = \frac{x_j}{p_j} \leq 1.$$
\end{proof}
For simplicity, we  assume wlog that $\sum_i \hat z_{ij} = 1$ and $\sum_j \hat z_{ij} = 1$, by adding  0 weight jobs.

\medskip

For the analysis, we use the following potential function.
    $$ \Phi(t) = \frac{2}{\eps}\left( f\Big(\bp^A(t)\Big) - \eps f \Big(\bp^A(t)  \ || \ (\frac{1}{\eps} - 1)\bp^O(t)\Big)\right)$$ 
Here we note that $f \Big(\bp^A(t)  \ || \ (\frac{1}{\eps} - 1)\bp^O(t)\Big)$
is the optimum objective of $\cL$ when we have  one copy of each job $j$ in $A(t)$ with remaining size $p^A_j(t)$ and $(1 / \eps - 1)$ copies of each job $j$ in $O(t)$, each with remaining size $p^O_j(t)$.

\subsubsection{Analysis Overview}

Before diving into the analysis, we first provide a high-level overview. Since the LP solution does not directly give us a time-indexed feasible schedule, we compromised true gradient descent. If possible, it would be the best for the algorithm to process exactly what the LP solution processes at time $\tilde t = 0$. Unfortunately as discussed before, this may not be feasible. 

To intuitively understand our algorithm, consider a single machine $i$ where jobs $1, 2, 3, \ldots$ are processed in this order, each by a half fraction. For simplicity, let us drop $i$ from the notation. Then, the LP solution on the machine is processing $j$ for $p_j / 2$ time steps. As mentioned before, fully processing job 1 on the machine may result in an infeasible schedule if job 1 is scheduled before other jobs on another machine in the LP solution. Our algorithm processes job 1 and job 2, each by half on the machine (these two jobs are the first unit fraction of jobs), to ensure a feasible schedule. This is surely not so effective as fully processing job 1 in decreasing $(*)$. However, for all jobs except 1 and 2, our algorithm is equally effective! Further, for the deficit due to jobs 1 and 2, $(**)$ comes to rescue. So, if $\tilde W_{(1)}$ is the total fractional weight processed by our algorithm, and we let $\tilde W_{(2)}$ be the other total fractional weight, our algorithm gets $\tilde W_{(1)}$ credit from $(**)$ and $\tilde W_{(2)}$ credit from $(*)$. Thus, $f(\bp^A(t)))$ decreases at a rate of at least $\tilde W_{(1)}$ + $\tilde W_{(2)}$, which is the total fractional weight of the remaining jobs in our algorithm's schedule. 

To show the algorithm is scalable, we create $(1/\eps - 1)$ copies of each job in the adversary's schedule and due to this we lose another $1 / \eps$ factor in the competitive ratio, but this is a minor technicality. The critical observation is that the approximate GD is as effective as GD with the aid of the decrease of $(**)$. 

As before we will consider the discrete changes and continuous changes separately.

\subsubsection{Analysis: Discrete Changes}

The proof of discrete changes of $\Phi$ is very similar to Lemma~\ref{lem:second-thm-super}, so we only highlight the differences. First, we will use generalized flow to show the supermodularity of our new residual objective. 

\begin{lemma}
    \label{lem:apx-ur-sup}
    $f(x)$ is (discrete-)supermodular.
\end{lemma}
\begin{proof}
    We follow the same strategy as we used to prove Lemma~\ref{lem:second-thm-super}. Consider the following valuation
    \begin{align}
        v_{\tilde t}(\bx(\tilde t)) := \max \sum_{i,j} \left(B \lambda_{ij} - \frac{w_j \tilde t}{p_{ij}} - 1 \right) z_{ij} \ \text{s.t.} \ 0 \leq \sum_i \lambda_{ij} z_{ij} \leq x_{j}(\tilde t), \sum_j z_{ij} \leq 1 \label{eqn:weighted-ur-sup}
    \end{align}
    for each time $\tilde t$. It is easy to see that the LP objective by flipping the sign and adding a big constant is equivalent to 
        $$\max \int_{\tilde t \geq 0} v_{\tilde t}(\bx(\tilde t)) \dd \tilde t \ \text{s.t.} \ \int_{\tilde t\geq 0} \bx(\tilde t) \dd \tilde t \leq \bx$$
    To show $f(x)$ is supermodular, it is sufficient to show that $v_{\tilde t}$ is LS for each time $\tilde t$. Toward this end, we can interpret the constraints in (\ref{eqn:weighted-ur-sup}) in the generalized flow setting as follows. For each machine $i$, node $i$ has excess capacity $1$. For each job $j$, node $j$ has excess capacity $x_j(\tilde t)$. Then add an arc from $j$ to $i$ with gain factor $1 / \lambda_{ij}$ for each job $j$ and machine $i$. Thus the constraints fall in the generalized flow polytope (\ref{eqn:p-gf}). Then we let the cost from $j$ to $i$ be $(B \lambda_{ij} - \frac{w_j \tilde t}{p_{ij}} - 1)$. From Theorem~\ref{thm:gf-ls}, we know $v_{\tilde t}$ is LS which leads to the lemma.
\end{proof}

For the sake of completeness, we give the proof of the following lemma that bounds the change of jobs' arrival or completion.

\begin{lemma}
    \label{lem:GD-l1-disc}
    The change of $\Phi(t)$ due to the job's completion or arrival is non-positive.
\end{lemma}
\begin{proof}
    The proof is almost identical to Lemma~\ref{lem:discrete-change-clairvoyant-int} although we use a slightly different potential function.  Thus we only point out the differences when a new job $j$ arrives. Then $f \big(\bp^A  \ || \ (\frac{1}{\eps} - 1)\bp^O\big)$ increases by at least: 
    \begin{align*}
        f(\bp^{A} \ || \ p_j \bone_{j^A} \ || \ (\frac{1}{\eps} - 1 )\bp^O \ || \ (\frac{1}{\eps} - 1 ) p_j  \bone_{j^O}) - f(\bp^{A} \ || \ (\frac{1}{\eps} - 1 ) \bp^O) \\
    =  g_{\by}(A \cup j^A \cup (\frac{1}{\eps} - 1 ) O \cup (\frac{1}{\eps} - 1 ) j^O) -  g_\by(A \cup (\frac{1}{\eps} - 1 ) O)
    \end{align*}
where $\by = \bp^{A} \ || \ p_j  \bone_{j^A} \ || \ (\frac{1}{\eps} - 1 ) \bp^O \ || \ (\frac{1}{\eps} - 1 ) p_j  \bone_{j^O}$. By supermodularity of $g$ and using the fact that $p^A(r_j) = p^O(r_j) = p_j$, we have 
\begin{align}
    &g_{\by}(A \cup j^A \cup (\frac{1}{\eps} - 1 ) O \cup (\frac{1}{\eps} - 1 ) j^O ) - g_\by(A \cup (\frac{1}{\eps} - 1 ) O) \nonumber \\
    & \geq  g_{\by}(A \cup j^A \cup (\frac{1}{\eps} - 1 ) j^O) - g_\by(A) \nonumber \\
    & =g_{\by}(A \cup ( 1/ \eps) j^A  ) - g_\by(A) \nonumber
    \end{align}
    Then we can decompose it to a telescopic sum and apply the supermodularity as follows.
    \begin{align*}
        &= \Big(g_{\by}(A \cup (1 / \eps) j^A) )- g_{\by}(A \cup (1 / \eps - 1) j^A ) \Big)+ ... + \Big(g_{\by}(A \cup j^A) - g_\by(A)\Big)  \\
        &\geq \frac{1}{\eps} ( g_{\by}(A \cup j^A) - g_\by(A)) 
    \end{align*}
    Therefore, $\Phi(t)$'s increase is at most $\frac{2}{\eps}(g_{\by}(A \cup j^A) - g_y(A) - \eps \cdot \frac{1}{\eps} ( g_{\by}(A \cup j^A) - g_\by(A)) ) = 0$, as desired. 
\end{proof}

\subsubsection{Analysis: Continuous Changes}

As before, we let $\tilde W^A(t) = \sum_j \frac{w_j}{p_j} p^A_j(t)$ be the total remaining fractional weight of our algorithm $A$ at time $t$. 
Then, let $\tilde W^A_{(1)}(t)$ be the fractional weight of jobs processed by the algorithm at $t$, i.e., $\tilde W^A_{(1)}(t) := \sum_{i,j} w_j \int_{\tilde t = 0}^{\tilde t_i} \frac{z_{ij}(\tilde t)}{p_{ij}} \dd \tilde t = \sum_{i,j} w_j \hat z_{ij}$. Let $\tilde W^A_{(2)}(t) := \tilde W^A(t) - \tilde W^A_{(1)}(t)$. For brevity, in all of the following lemmas, we omit that we consider a time $t$ when there is no jobs' completion and arrival.
The following lemma shows that the approximate GD still decreases the new residual optimum effectively.

\begin{lemma}
    \label{lem:GD-l1-A}
    $\frac{\dd}{\dd t} f(\bp^A(t)) \leq -s \tilde W^A(t)$.
\end{lemma}
\begin{proof}
    Fix any time $t$ and consider the change during time $[t, t + \dd t)$. Let $\bz$ be the optimum solution of $f(\bp^A(t))$, i.e, $\bz = \arg\min_{\bz'} \cL_{\bp^A(t)}(\bz')$. 
    Recall that our algorithm schedules $\hat \bz$ during time interval $[t, t + \dd t)$ and it is given $s$-speed, i.e., GD processes a job $j$ on machine $i$ by $s \dd t \cdot \hat z_{ij}$ fraction, where $\hat z_{ij} := \int_{\tilde t=0}^{\tilde t_i} \frac{z_{ij}(\tilde t)}{p_{ij}} \dd \tilde t$. Now following the same strategy in Lemma~\ref{lem:second-theorem-gd}, we need to find a feasible schedule $\bz'$ (with $1$-speed) for remaining job sizes $\bp^A(t) - s \dd t \hat \bZ$ where $\hat Z_j := \sum_i \lambda_{ij} \hat z_{ij}$. Then, the construction of $\bz'$ is as follows. First, we remove $s \dd t \cdot z_{ij}(\tilde t) / p_{ij}$ fraction from $\bz$ at local time $\tilde t \leq \tilde t_i$ for job $j$ on machine $i$. Note that this empties out $s \dd t$ units of space on the machine because $\sum_j \int_{\tilde t=0}^{\tilde t_i} \frac{z_{ij}(\tilde t)}{p_{ij}} \dd \tilde t = 1$. We obtain $\bz'$ by removing the empty spaces on each machine. 
    
    Let's fix machine $i$ and consider the change of $(*)$ in the residual LP (\ref{eqn:lp-um-l1}). For any $\tilde t > \tilde t_i$, since $s \dd t \cdot \sum_j \hat z_{ij} = s \dd t$ fraction of jobs completed before time $\tilde t_i$, the contribution of $z_{ij}(\tilde t)$ in $\cL_{\bp^A(t + \dd t)}(\bz')$ becomes at most $\frac{w_j}{p_{ij}}(\tilde t - s \dd t) z_{ij}(\tilde t)$.

    Thus, the decrease of $f(\bp^A(t))$ from $(*)$ is by at least 
    \begin{align*}
        \sum_{i,j} \frac{w_j}{p_{ij}} \int_{\tilde t \geq \tilde t_i} \Big(\tilde t - (\tilde t - s \dd t )\Big) z_{ij}(\tilde t) \dd \tilde t = s \dd t  \cdot \sum_{i,j} \frac{w_j}{p_{ij}} \int_{\tilde t \geq \tilde t_i} z_{ij}(\tilde t) \dd \tilde t = s\dd t \cdot \tilde W^A_{(2)}(t)
    \end{align*}
    Since job $j$ gets processed by $s \dd t \cdot \hat z_{ij}$ fraction on machine $i$, the change of $f(\bp^A(t))$ from $(**)$ is exactly 
    \begin{align*}
        \sum_{i,j} w_j \Big( \int_{\tilde t \geq 0} z_{ij}(\tilde t) \dd \tilde t - \int_{\tilde t \geq 0} z'_{ij}(\tilde t) \dd \tilde t \Big) &= \sum_{i,j} w_j \cdot s \dd t  \cdot \hat z_{ij} = s\dd t \cdot \tilde W^A_{(1)}(t)
    \end{align*}
    Therefore, the lemma follows by putting two  bounds together.
\end{proof}

In the following $|_{GD}$ and $|_{\opt}$ mean that we consider the change due to GD's processing and $\opt$'s processing alone respectively. 

\begin{lemma}
    \label{lem:GD-l1-AO}
$-\frac{\dd}{\dd t} f \big(\bp^A(t) \ || \ (\frac{1}{\eps} - 1)\bp^O(t)\big) \Big |_{GD} \leq 2s \tilde W^A(t) + s(\frac{1}{\eps} - 1) \tilde W^O(t)$.
\end{lemma}

\begin{proof}
    Fix a time $t$ and consider the change during $[t, t + \dd t)$, where no job arrives or completes assuming that only GD processes jobs. For notational simplicity, let $\bx(t):= \bp^A(t) \ || \ (\frac{1}{\eps} - 1)\bp^O(t)$. To achieve our goal of upper bounding the following quantity,
        $$ f(\bx(t)) - f(\bx(t + \dd t)),$$
    we construct a feasible solution $\bz''$ to $\cL_{\bx(t)}$. Let $\bz'$ be an optimum solution of $f(\bx(t + \dd t))$, i.e, $\bz' = \arg\min_{\bz'''} \cL_{\bx(t + \dd t)}(\bz''')$. Let $\bz''$ be the time-indexed schedule which first has $\hat z_{ij}$ for $s \dd t$ units of time and then exactly follow $\bz'$ on each machine $i$. Clearly, we have 
        $$ f(\bx(t)) - f(\bx(t + \dd t)) \leq  \cL_{\bx(t)} (\bz'') - f(\bx(t + \dd t))$$

    Let's first consider the change of $(**)$ in the LP objective. Since GD processes jobs according to the distribution $\{\hat z_{ij}\}_j$ during $[t, t+\dd t)$, $(**)$ decreases by 
    $$
        \sum_{i,j} w_j \hat z_{ij} s \dd t = \tilde W^A_{(1)} s \dd t,
    $$
    as we observed in the proof of Lemma~\ref{lem:GD-l1-A}.
    
    We now turn our attention to the change of $(*) $, which is the fractional weighted completion time of jobs. By scheduling $\hat z$ with speed $s$, every terms coefficient changes from $w_j/ p_{ij} \tilde t$ to $w_j/ p_{ij} (\tilde t - s \dd t)$. Thus, the decrease is at most the total fractional weight of all jobs times $s \dd t$, which is 
     $$       (\tilde W^A(t) + (1 / \eps - 1) \tilde W^O(t)) s \dd t$$
    By adding up the upper bounds on the change of $(*)$ and $(**)$ and using the fact that $\tilde W_{(1)}^A(t) \leq \tilde W^A(t)$, we obtain the lemma. 
\end{proof}

We now consider the change due to the adversary's processing.
Note that  the RHS consists of $\tilde W^A(t) := \sum_{j \in A(t)} w_j\frac{p^A_j(t)}{p_j}$ and $ W^O(t) := \sum_{j \in O(t)} w_j$, which denote the  total fractional remaining weight of jobs in $A(t)$ and the total integral weight of jobs in $O(t)$, respectively.

\begin{lemma}
    \label{lem:GD-l1-AO-2}
    $-\frac{\dd}{\dd t} f \big(\bp^A(t) \ || \ (\frac{1}{\eps} - 1)\bp^O(t)\big) \Big |_{\opt} \leq (1/ \eps - 1) \tilde W^A(t) + (1/ \eps) (1/ \eps - 1)  W^O(t)$.
\end{lemma}
\begin{proof}
    The proof of this lemma is very similar to 
    that of Lemma~\ref{lem:GD-l1-AO}.
    We keep the same notation for $\bx$ and $\bz'$. Since the adversary processing job $j$ with 1 speed is  equivalent to processing all $(1 / \eps - 1)$ copies of the same job at the same rate, we will pretend that all the copies are distinct jobs and the adversary has $(1/ \eps -1)$-speed. Clearly, this only gives more power to the adversary. 
    For simplicity we assume that the adversary schedules an integral matching during $[t, t + \dd t)$, since the extension of the analysis to fractional matching is straightforward. That is, the adversary processes exactly one job $j(i)$ on machine $i$ during $[t, t + \dd t)$ and processes it on no other machines. Let $\hat z^*$  denote this schedule. So, $z^*_{i, j(i)} = 1$ but $z^*_{i, j'} = 0$ for all $j' \neq j(i)$.    
    Let $\bz''$ be a schedule that has $\hat z^*$ for $(1/ \eps - 1) \dd t$ time steps, followed by $\bz'$.

    We note that $(**)$ decreases by 
    $$
        \sum_{i,j} w_j \hat z^*_{ij} (1 / \eps -1) \dd t = \sum_i w_{j(i)} (1 / \eps -1) \dd t \leq (1 / \eps -1) W^O(t)  \dd t,
    $$
 where the last inequality follows from the fact that $j(i) \neq j(i')$ for all $i \neq i'$. 
    
    By observing that the total fractional weight 
    in $\bp^A(t) \ || \ (\frac{1}{\eps} - 1)\bp^O(t)$ is $\tilde W^A(t) + (1 / \eps - 1)\tilde W^O(t)$, we can upper bound the change of $(*)$ by 
    $$(1/ \eps - 1) \dd t (\tilde W^A(t) + (1 / \eps - 1)\tilde W^O(t))$$
    We obtain the lemma by adding the two upper bounds and using the fact that $\tilde W^O(t) \leq W^O(t)$.    
\end{proof}
    
\begin{corollary}
    \label{lem:GD-l1-Phi}
    Consider any time step $t$ that no arrival or completion of jobs occurs. For $\eps < \frac{1}{8}$ and $A$ is given $s = 1 + 2\eps$-speed, we have $\frac{\dd}{\dd t}\Phi(t) \leq - \tilde W^A(t) + O(1 / \eps^2) \tilde W^O(t)$.
\end{corollary}
\begin{proof}
    From Lemma~\ref{lem:GD-l1-A}, Lemma~\ref{lem:GD-l1-AO}, and Lemma~\ref{lem:GD-l1-AO-2}, we have 
    \begin{align*}
        \frac{\dd}{\dd t} \Phi(t) &\leq \frac{2}{\eps} \Big( -s \tilde W^A(t) + \eps \big( (2s + \frac{1}{\eps} - 1) \tilde W^A(t) +  (s + \frac{1}{\eps})(\frac{1}{\eps} - 1) W^O(t) \big) \Big) \\
        &= \frac{2}{\eps} \Big( (-\eps + 4\eps^2) \tilde W^A(t) +  O(1 / \eps) W^O(t) )\Big)\\
        &\leq -\tilde W^A(t) + O(1 / \eps^2) W^O(t)
    \end{align*}
\end{proof}

By Lemma~\ref{lem:GD-l1-disc} and $\Phi(T) = \Phi(0) = 0$, we have $\int_{t \geq 0} \frac{\dd}{\dd t} \Phi(t) \geq 0$. By Lemma~\ref{lem:GD-l1-Phi}, we have $\int_{t \geq 0} \tilde W^A(t) \leq O(1 / \eps^2) \int_{t \geq 0} W^O(t)$. LHS is the total fractional weighted flow time of the algorithm and RHS is the total integral weighted flow time of the optimum schedule. Therefore, by Lemma~\ref{thm:conversion}, we obtain a $(1+\eps)$-speed, $O(1/\eps^3)$-competitive algorithm for the integral objective.

\subsection{Unrelated Machines via Immediate Dispatch}
    \label{sec:unrelated-immediate}

In this section, we reproduce the result in \cite{AnandGK12}, which gave a $(1+\eps)$-speed $O(1/ \eps)$-competitive algorithm for minimizing total weighted flow time. The analysis of \cite{AnandGK12} was based on dual fitting unlike \cite{ChadhaGKM09}, which gave a $(1+\eps)$-speed $O(1/ \eps^2)$-competitive algorithm using potential functions. We show that potential functions can also achieve an $O(1/\eps)$-competitive ratio. At the high level, \cite{AnandGK12} considered integral residual optimum unlike \cite{ChadhaGKM09} that used fractional residual optimum.

Since we will use exactly the same algorithm as \cite{AnandGK12} and obtain the same competitive ratio of $O(1 /\eps)$, we will only sketch the analysis, focusing on how our approach leverages supermodularity and gradient descent explicitly.

The algorithm we consider is non-migratory and immediate-dispatch. Let $A_i(t)$ be the jobs assigned to machine $i$. Note that $\{A_i(t)\}_i$ will be a partition of $A(t)$. Let $\bp^{A_i(t)}$ be the (vector of the) remaining sizes of jobs assigned to  machine $i$ at time $t$ in our algorithm and define $\bp^{O_i(t)}$ analogously for the optimum schedule. 

We define $f(\bp^{A_i}(t))$ as the residual optimum at time $t$ of the jobs assigned to machine $i$. Formally,  letting $p_{ij}(t) = p_{j}(t) / \lambda_{ij}$, $f(\bp^{A_i}(t)) = \sum_{j} w_j \sum_{j' \leq j \textnormal{ s.t. } w_{j'} / p_{ij'}(t) \geq w_j /p_{ij}(t) }p_{ij'}(t) $ is the minimum total weighted completion time of jobs in $A_i(t)$ pretending that each job $j \in A_i(t)$ has a remaining size $p_j(t)$, weight $w_j$, arrival time $0$; ties are broken in an arbitrary but fixed order. Since $f$ is concerned with the single machine scheduling, it is straightforward to show its supermodularity. Then we let $\bar f(\bp^A(t)) := \sum_i f(\bp^{A_i}(t))$. 

We can restate \cite{AnandGK12} algorithm as follows: When a job $j$ arrives at time $t$, we permanently assign the job $j$ to the machine $i$ that incurs the minimum increment of the residual optimum $\bar f(\bp^A(t^-))$. Then, we process the jobs on each machine using gradient descent, i.e. follow the residual optimum schedule on the machine.

For the analysis, we use the following potential function, where $\bar f(\bp^A(t) \ || \ \bp^O(t))$ is defined analogously for jobs in $A$ and $O$ together. 
    $$ \Phi(t) = \frac{2}{\eps}\Big(\bar f(\bp^A(t)) - \frac{1}{2} \bar f(\bp^A(t) \ || \ \bp^O(t) )\Big) $$

    While we only consider the non-migratory optimum for brevity, it can be extended to the migratory case, following the approach described in \cite{ChadhaGKM09}.

\paragraph{Continuous Changes.} If we focus on machine $i$, we only need to consider the change of $\frac{2}{\eps}( f(\bp^{A_i}(t)) - \frac{1}{2}  f(\bp^{A_i}(t) \ || \ \bp^{O_i}(t) ))$. This is exactly the same as the single machine scheduling where we can pretend that we only have jobs in $A_i(t)$ and the adversary has jobs in $O_i(t)$. Thus, 
using Corollary~\ref{cor:running-clairvoyant-int}, we have that 
$\frac{\dd}{\dd t} \Phi(t)$ restricted to machine $i$ is at most $ - \sum_{j \in A_i(t)} w_j + (1 + \frac{2}{\eps}) \sum_{j \in O_i(t)} w_j$. Summing over all machines we have the following lemma. 
\begin{lemma}
    \label{lem:unrelated-id-continuous}
    Consider any time $t$ when no job arrives or is completed by our algorithm or the adversary. When GD is given $1+\eps$-speed, $\frac{\dd}{\dd t} \Phi(t) \leq -W^A(t) + (1 + \frac{2}{\eps}) W^O(t)$.
\end{lemma}

\paragraph{Discrete Changes.}
The following lemma bounds the potential change due to a job's arrival or completion. Unlike the proofs in the previous section, we take a closer look at the machine to which the job is assigned. 

\begin{lemma}
    \label{lem:unrelated-id-discrete}
    If $\Phi(t)$ changes discontinuously due to a job's arrival or completion, the change is non-positive. 
\end{lemma}
\begin{proof}
    For the same reason we used in the proof of Lemma~\ref{lem:discrete-change-clairvoyant-int}, we only need to focus on the arrival case. Suppose a new job $j$ arrives at time $t$. We use the same notations in Lemma~\ref{lem:discrete-change-clairvoyant-int}. So $j^A$ and $j^O$ are job $j$'s copy in the schedule of $A$ and $O$ respectively. We assume wlog that job $j^A$ is assigned to machine $1$ by the algorithm. If $j^O$ is assigned to the same machine $1$, the change of $\bar f(\bp^A \ || \ \bp^O)$ is at least
    \begin{equation}
        f(\bp^{A_1} \ || \ p_{1j} \bone_{j^A} \ || \ \bp^{O_1} \ || \ p_{1j} \bone_{j^O}) - f(\bp^{A_1} \ || \ \bp^{O_1}) \geq 2 \left(f(\bp^{A_1} \ || \ p_{1j} \bone_{j^A}) - f(\bp^{A_1}) \right) 
    \end{equation}
    where the inequality holds due to the supermodularity of $f$.
    
    Suppose $j^O$ is assigned to another machine, say machine $2$. Since the change only occurs on machines $1$ and $2$, the increase of $\bar f(\bp^A \ || \ \bp^O)$ is at least
    \begin{align}
        &f(\bp^{A_1} \ || \ p_{1j} \bone_{j^A} \ || \ \bp^{O_1}) - f(\bp^{A_1} \ || \ \bp^{O_1}) + f(\bp^{A_2} \ || \ \bp^{O_2} \ || \ p_{2j} \bone_{j^O}) - f(\bp^{A_2} \ || \ \bp^{O_2})  \nonumber \\        
        \geq \; & f(\bp^{A_1} \ || \ p_{1j} \bone_{j^A}) - f(\bp^{A_1}) + f(\bp^{A_2} \ || \ p_{2j} \bone_{j^O}) - f(\bp^{A_2}) \label{eq:unrelated-id-eq2}, 
    \end{align}
    where the inequality follows from $f$'s supermodularity.
    
    Recall that the algorithm assigns job $j$ to the machine $i$ that gives the minimum increment of $f$. Thus we have 
    \[
        f(\bp^{A_1} \ || \ p_{1j} \bone_{j^A}) - f(\bp^{A_1}) \leq f(\bp^{A_2} \ || \ p_{2j} \bone_{j^O}) - f(\bp^{A_2})
    \]
    Then, Eqn.(\ref{eq:unrelated-id-eq2}) is at least
       $ 2 \left(f(\bp_1^A \ || \ p_{1j} \bone_{j^A}) - f(\bp_1^A) \right) $.
    Thus, in both cases, we have shown that $\bar f(\bp^A \ || \ \bp^O)$ increases by at least $ 2 \left(f(\bp_1^A \ || \ p_{1j} \bone_{j^A}) - f(\bp_1^A) \right) $. Since the change of $\bar f(\bp^A)$ is exactly
        \[ f(\bp^{A_1} \ || \ p_{1j} \bone_{j^A}) - f(\bp^{A_1}), \]
    we can conclude that  $\Phi(t)$'s change due to job $j$'s arrival is non-positive.         
\end{proof}

We combine Lemmas \ref{lem:unrelated-id-continuous} and 
\ref{lem:unrelated-id-discrete} as we combined 
Lemma~\ref{lem:discrete-change-clairvoyant-int} and  Corollary~\ref{cor:running-clairvoyant-int} at the end of Section~\ref{sec:proof-thm:main-clairvoyant-int}. As a result, we can show that the online algorithm is $(1+\eps)$-speed $O(1/\eps)$-competitive  for the total weighted flow time objective. 
\section{Other Related Work}
    \label{sec:other-related}

Gross substitutes valuations were introduced by Kelso and Crawford in economics \cite{KelC82}. Interestingly, the same valuations were introduced in other areas under different names, such as M$^\natural$-concave functions \cite{murota1999m}, matroidal maps \cite{dress1995well,dress1995rewarding}, and valuated matroids \cite{dress1992valuated}. Thus, there are several equivalent ways to characterize them. The interested reader is referred to the nice survey  \cite{leme2017gross} and tutorial talk given in ACM EC 2018 \cite{gstutorial}, and the extensive survey \cite{murota2022discrete}. 

We discuss two assumptions we make in this paper, clairvoyance and free preemption. Clairvoyant algorithms are assumed to know a job's size upon its arrival---in contrast, non-clairvoyant algorithms do not until completing the job. See~\cite{motwani1994nonclairvoyant,KalyanasundaramP00,edmonds2003non,PruhsST04,EdmondsP12,im2014focs,ImKM18,DBLP:conf/icalp/0001G0019} for some examples of non-clairvoyant algorithms. There have recently been work to design online algorithms when we are given jobs sizes that are not completely accurate \cite{purohit2018improving,im2021non,azar2021flow,AzarLT22,LindermayrM22}, where the  goal is to obtain a schedule almost as good as what the best clairvoyant algorithm can offer when the (predicted) sizes are not far from the actual sizes. In this paper we allow online algorithms to preempt jobs for free. In fact, the flow time minimization becomes computationally very hard without preemption. Thus, non-preemptive scheduling has been studied with resource augmentation even in the offline setting \cite{BansalCKPSS07,ChuzhoyGKN04,ImLMT15}. 

In this work we assume there are no precedence constraints among jobs. For online precedence constrained scheduling for total (weighted) flow time, see \cite{AgrawalLLM16,DBLP:conf/icalp/0001G0019}.

There exists a large literature on the ``delay" variants of online problems, where the online algorithm can procrastinate its decision incurring delay cost, e.g., \cite{azar2017online,azar2018set,azar2020beyond}. The main difference is that in our problem, we have capacitated resources that we can use to process jobs at each time, while the online algorithms pay extra cost for taking an action such as buying a set in their setting. 

Deng \etal \cite{DengHL07} study Walrasian equilibrium for a market where jobs act as agents bidding for CPU times. Their work considers the single-machine setting, where jobs have different valuations depending on their own completion time. The focus is on the existence and computability of the equilibrium.

\appendix

\section{Missing Proofs for Valuation Functions}
    \label{sec:ls-proofs}

\subsection{Proof of Lemma~\ref{lem:LS-submodular}}
\begin{proof}
    We have  $v(\bx) = \min_{\bq\geq 0} \set{ \pi(\bq) + \bq \bx}$ (See \cite{milgrom2009substitute} Theorem 26, the concavity implies this). Let $f(\bx, \bq) = \pi(\bq) + \bq \cdot \bx$. We can show $\bq \cdot \bx$ is supermodular in $\bq$ and $\bx$ when $\bq, \bx \geq 0$ using on each coordinate the fact that $a b + a' b' \leq \max\{a, a'\} \max\{b, b'\} + \min\{a, a'\} \min\{b, b'\}$ for any scalars $a, a', b, b' \geq 0$.
    Then, we have $f(-\bx, \bq) = \pi(\bq) - \bq \cdot \bx$, which is submodular in $\bx$ and $\bq$. By taking the min over all $\bq$, we know $v(-\bx)$ is submodular for the following reason. 
    Let $g(\bx, \bq) := f(-\bx, \bq)$. Let $\bq$ and $\bq'$ be such that $v(-\bx) = g(\bx, \bq)$ and $v(-\bx') = g(\bx', \bq')$. We then have, 
    \begin{align*}
        v(-\bx) + v(-\bx') &= 
        g(\bx, \bq) + g(\bx', \bq')  \\
        & \geq g( \bx \vee \bx', \bq \vee \bq') + g( \bx \wedge \bx', \bq \wedge \bq') \quad \mbox{[$g(\bx, \bq)$ is submodular in both $\bx, \bq$]}\\        
        & \geq v( - (\bx \vee \bx')) + v( -\bx \wedge \bx'),
    \end{align*}
    which proves $v(-\bx)$'s submodularity  in $\bx$. 
    So, we have 
    \begin{align*}
v(\bx) + v(\by) &= v(- (-\bx)) + v((-\by)) \\
        &\geq v(- (-\bx \vee-\by) + v(-(-\bx \wedge -\by)) \quad \mbox{[$v(-\bx)$ is submodular in $\bx$]}\\
        &= v( \bx \wedge \by) + v(\bx \wedge \by)
    \end{align*}
\end{proof}

\subsection{Concave Closure of GS: Proof of Theorem~\ref{thm:gs-concave-closure}} 

Consider a fixed gross-substitute valuation $v:\set{0, 1}^n \to \R$.\footnote{We interchangeably use $v(S)$ and $v(\bone_S)$} We first observe that for any rational number $\delta > 0$, we can approximate $v$ by a function $v':\set{0, 1}^n \to \Q$ such that $|v(S) - v'(S)| \leq \delta$ for any $S \subseteq [n]$. This is because due to \cite{gstutorial} that shows that $v$ is GS if and only if it satisfies a finite number of linear inequalities, we can set up a linear programming with variables $\{v'(S)\}_{S \subseteq [n]}$ satisfying the inequalities together with the approximation requirement for all $S$. The LP has a feasible solution, and particularly a rational number solution since $\delta \in \Q$. Thus, we assume $v:\set{0, 1}^n \to \Q$ henceforth.

We show that for any $\bx \in \Q^n \cap [0, 1]^n$, we can view $v^+(\bx)$ as a convolution of multiples copies of $v$. The proof is not difficult, but we include it for completeness.

\begin{lemma}
    \label{lem:GS-equal}
    Define $\tilde v^{+} : \Q^n \cap [0, 1]^n \to \Q$ such that
    $$\tilde v^{+}(\bx) = \sup_{\eps = 1 / k! : k \in \Z_{\geq 0} } \eps \max_{ \bx / \eps = \bx_1 + ... + \bx_{1 / \eps}} \sum_{i = 1}^{ 1 / \eps } v(\bx_i)$$
    where $\bx_1, ..., \bx_{1 /\eps} \in \set{0, 1}^n$. Then, for any $\bx \in \Q^n \cap [0, 1]^n$, 
    $$ \tilde v^{+}(\bx) := v^+(\bx) $$ 
    where $ v^+(\bx) = \max \set{\sum_S v(S) \lambda_S : \sum_{S} \lambda_S \bone_S = \bx, \sum_{S} \lambda_S = 1, \lambda_S \geq 0}$ is the concave closure of $v$. 
\end{lemma}
\begin{proof}
    Fix a $\bx \in \Q^n \cap [0, 1]^n$. Let $\set{\lambda_S}_S$ be an optimal solution of $v^+(\bx)$; note that $\lambda_S \in \Q$ for all $S$. Consider an $\eps$ such that $\lambda_S$ is an integer multiple of $\eps$ for all $S$. To show $\tilde v^+ (\bx) \geq v^+(\bx)$, we assign $\bone_S$ to $\lambda_{S}  / \eps$ copies of $v$. Then, 
    $\tilde v^+ (\bx)  \geq \eps \sum_S (\lambda_S / \eps) v(\bone_S) = v^+ (\bx)$.
    Conversely, fix an $\eps > 0$ such that $1 / \eps$ is an integer, together with any $\bx_1, \bx_2, \ldots, \bx_{1 / \eps} $ such that $\bx / \eps = \sum_{i=1}^{1 / \eps} \bx_i$. Then, by setting $\lambda_{S} = \eps |\set{\bx_i = \bone_S}_i|$, one can immediately show $v^+(\bx) \geq \tilde v^{+}(\bx)$. 
    Finally, it is easy to see that $\tilde v(\bx)$ is well defined as in the definition of $\tilde v^+(\bx)$, the quantity monotonically increases as $k \to \infty$. 
\end{proof}

\begin{remark}
    Strictly speaking, when we create $1/ \eps$ copies of $v$, we also create $1 / \eps$ copies of each item. Each copy of $v$ is assumed to have no benefit from having more than one copy of any item. It is easy to check that  this valuation function is GS over all the item copies.
\end{remark}

The following corollary is immediate from the proof of the above lemma. 
\begin{corollary}
    For any $\bx \in \Q^n \cap [0, 1]^n$, there exists $\eps > 0$ such that $v^{+\eps}(\bx) = \tilde v^{+}(\bx) = v^+(\bx)$, where $$v^{+\eps}(\bx) =  \eps \max_{ \bx / \eps = \bx_1 + ... + \bx_{1 / \eps}} \sum_{i = 1}^{ 1 / \eps } v(\bx_i)$$
\end{corollary}

From Lemma~\ref{lem: LS-equal}, it is sufficient to show the dual profit of $v^+$ is submodular. 

\begin{lemma}
    \label{lem: GS-eps-Sub}
    The dual profit $\pi^{+}(\bq) = \max_{\bx \in \R^n} v^{+}(\bx) - \bq \cdot \bx$ is submodular.
\end{lemma}
\begin{proof}
    We first assume $\bq \in \Q^n$. When $\bq \in \Q^n$, we know $\bx = \arg \max_{\bx' \in \R^n} v^{+}(\bx') - \bq \cdot \bx'$ is in $\Q^n$. Thus, from the previous lemma and corollary, we have $v^+(\bx) = \tilde v^+(\bx)$. Further, for some $\eps > 0$, we have $v^{+\eps}(\bx) = v^+(\bx)$. Then, we have,
    \begin{align*}
        \pi^{+}(\bq) &= \max_{\bx \in \Q_+^n} v^{+\eps}(\bx) - \bq \cdot \bx \\ 
        &= \max_{\bx \in \Q_+^n} \eps \max_{ \bx / \eps  = \bx_1 + ... + \bx_{1 / \eps }}  \sum_{i = 1}^{ 1 / \eps} v(\bx_i) - \bq \cdot  \bx_i \\
        &= \eps \max_{\bx_1, ..., \bx_{1/\eps} \in \set{0, 1}^n } \sum_{i = 1}^{1 / \eps} v(\bx_i) - \bq \cdot  \bx_i \\
        &= \eps ( 1/ \eps) \pi(\bq) \\
        &= \pi(\bq)
    \end{align*}
     Thus, we have shown that $\pi^+$ coincides with $\pi$ on all rational vectors $\bq$. Due to continuity, it does on all real vectors $\bq$. Since $v$ is a gross-substitute valuation, $\pi$ is submodular. Therefore, $\pi^+$ is submodular, as desired. 
\end{proof}

\subsection{Generalized Flow is Linear Substitute}

To show that $v$ is LS, from Lemma~\ref{lem: LS-equal}, we know that it is equivalent to showing its dual profit $\pi(\bq)$ is submodular w.r.t price vector $\bq \in \R^n$. The following lemma is due to 
Fleischer \cite{Fle10} based on LP primal-dual. For a combinatorial proof, see~\cite{li2011generalized}.

\begin{lemma} [Fleischer \cite{Fle10}]
    \label{lem:Fle-supmodular}
    Let $g(S)$ denote the (cost of the) min-cost flow with non-negative excess at the set of nodes $S \subseteq V \setminus V^-$, that is the min-cost flow with excess vector $b'$ defined as $b'_v = b_v$ for $v \in S \cup V^-$; and $b'_v = 0$ for the other nodes. Then, $g$ is a supermodular.
\end{lemma}

Intuitively, $g$ measures the cost of the min-cost flow of the underlying network $G$ when we `switch' on a subset $S$ of sources $V^{-}$; all the other sources do not generate any flow. As we add more sources, they create more congestion, which causes supermodularity. 

We show that Lemma~\ref{lem:Fle-supmodular} implies the valuation function we consider is LS. The proof idea is relatively simple. Due to Lemma~\ref{lem: LS-equal}, it suffices to show the dual profit function is submodular. Towards that end, we simulate increasing the price of a job's processing unit by adding/activating new source nodes associated with the job.

We first show submodularity when the prices are rational numbers, i.e., $\bq \in \eps \Z^n$, for any $\eps > 0$ such that $1 / \eps$ is an integer, and extend it to arbitrary prices.

\begin{lemma}
    \label{lem: GF-eps-Sub}
    For any $\eps \in 1 / \Z_+$, $\pi^{+\eps}(\bq) = \max_{\bx \in \R_+^n} v(\bx) - \bq \cdot \bx$ is submodular where $\bq \in \eps \Z_+^n$.
\end{lemma}

\begin{proof}
     We construct an auxiliary network $\tilde G$ such that the min-cost flow of $\tilde G$ captures the dual profit $-\pi^{\eps}$ as follows.
    We extend $G$ to  $\tilde G$ by adding nodes and arcs as follows. Without declaration, the flow gain $\gamma_e$ is $1$ and the capacity $u_e$ is $\infty$ for all new arcs $e$, and the excess $b_v$ are $\infty$ for all new nodes. Consider the source node for job $j$ in $G$. Let $q_{j}$ be its price. First, add a new node $j_{0}$ to $G$ and connect it to $j$ by an arc $(j_{0},j)$ with cost $- c_{j}$. Then we add a series of  arcs: for every $j \in \Z_+$, we add a new node $j_i$, and connect it to $j_{i-1}$ with cost $\eps$. Finally, we set the cost of each arc in $G$ to be 0.

    Then we let $g(S)$ denote the min-cost flow with nonnegative excess at the set of nodes $S$, that is the excess $b'$ defined as $b_v' = b_v$ for all $v \in S \cup V^-$; and $b_v' = 0$ for the other nodes. Now fix any price vector $\bq \in \eps \Z_+^n$. Let $S_{\bq}$ be a set of nodes that contains $\set{j_{q_j / \eps}, j_{q_j / \eps + 1}, ...}$ for each $j \in J$. Then we notice that $-\pi^{+\eps}(\bq)$ is equal to $g(S_{\bq})$. Consider two price vectors $\bq, \bq' \in \eps \Z_+^n$ and $\bq \leq \bq'$, and a job $j$. Note that $\pi^{+\eps}(\bq + \eps \bone_j) - \pi^{+\eps}(\bq) = - g(S_{\bq + \eps \bone_j}) + g(S_{\bq})$ and $S_{\bq + \eps \bone_j} + \set{j_{q_j / \eps}} = S_{\bq}$. Because $\bq \leq \bq'$, we have $S_{\bq'} \subseteq S_{\bq}$. Then the lemma follows from the supermodularity of $g$ due to Lemma~\ref{lem:Fle-supmodular}. \qedhere 
\end{proof}

\begin{lemma}
    The dual profit $\pi : \R_+^n \to \R$ is submodular.
\end{lemma}

\begin{proof}
    For any feasible $\bx$ of $v$, we can write $\pi(\bq) = \max_\bx (\bc - \bq) \cdot \bx$. Then, it is clear that $\pi$ is continuous.
    
    Suppose that $\pi$ is not submodular. Let $\bq, \bb \in \R_+^n$ be the price vectors that violate the submodularity. Then, we have 
    $$ \pi(\bq) + \pi(\bb) < \pi(\bq \vee \bb) + \pi(\bq \wedge \bb) $$
    Let $\Delta$ be the difference of the above inequality. 
    From the continuity of $\pi$, for some infinitesimally small $\eps' > 0$,  we have $|\pi(\bq) - \pi(\bq')|, |\pi(\bb) - \pi(\bb')|, |\pi(\bq \vee \bb) - \pi(\bq' \vee \bb')| \ \text{and} \ |\pi(\bq' \wedge \bb) - \pi(\bq' \wedge \bb')|$ are less than $\eps'$, for some $\delta$ such that $\norm{\bq - \bq'}_1 < \delta$ and $\norm{\bb - \bb'}_1 < \delta$ where $\bq',\bb' \in \eps \Z_+^n$ for some $\eps \in 1 / \Z_+$. Since $\Delta - 4 \eps' > 0$, we have 
    $$\pi^{+\eps}(\bq') + \pi^{+\eps}(\bb') < \pi^{+\eps}(\bq' \vee \bb') + \pi^{+\eps}(\bq' \wedge \bb')$$  
    This would contradict the submodularity of $\pi^{+\eps}$ from Lemma~\ref{lem: GF-eps-Sub}, and thus we have the lemma. \qedhere
\end{proof}
\section{Gradient Descent is Not a Panacea}
    \label{sec:gd-bad-bdcast}

In this section, we briefly show that gradient descent is not $O(1)$-competitive with any $O(1)$-speed. In the (pull-based) broadcast scheduling \cite{aksoy1998scheduling,kalyanasundaram2000scheduling,bansal2005approximating,bansal2010better,im2012talg}, the server stores some $n$ pages of useful information. At each time some jobs arrive, each asking for a page. For simplicity, assume all pages are unit-sized. When the server broadcasts a page, it satisfies all jobs requesting the same page and it can broadcast only one page at a time.
Scalable algorithms are known for this problem \cite{bansal2010better,im2012talg}.

Suppose gradient descent has $s$-speed and $s$ is an integer. Consider the following instance. There are special pages, $e_1, e_2, \ldots, e_s$. At each time, a total of $2s$ jobs arrive, two jobs asking for each special page. Further, in every $s+1$ time steps, one job arrives asking for a distinct page. Suppose the instance ends at time $(s+1)T$. At each time, gradient descent processes the $2s$ jobs requesting the special pages, having no time to work on any other jobs. Thus, at time $T$, it has $\Omega(T)$ jobs alive, which have $\Omega(T)$ flow time on average. Therefore, gradient descent's total flow time is $\Omega(T^2)$. On the other hand, the adversary can repeat the following: Broadcast each of the special $s$ pages and then the unique other page requested. It is easy to see that it satisfies every job within $O(1)$-time step. Therefore, its total flow time is $O(T)$.

\section{Making Gradient Descent Run in Polynomial Time}
    \label{sec:rt}

The high-level idea is to replace time slots with intervals of exponentially increasing lengths to reduce the number of variables to consider. This idea is commonly used in approximation algorithms to make  a time-indexed LP compact.  While it is straightforward to see it only losing $1+\eps$ factor in the approximation ratio, here we should be more careful as we need to show the change of the approximate LP optimum (residual optimum) is $1+\eps$-approximate as opposed to what we obtain with the exact residual optimum. 

We define time intervals $\cI := \{ I_h := [a_h, a_{h+1}) \; | \; h \geq 1\}$, where $a_1 < a_2 < \ldots $ are the integers in $[\frac{10}{\rho^2}] \cup \{  \lfloor \frac{10}{\rho^2} (1+\rho)^l \rfloor \; | \; l \geq 1  \}$, in increasing order, for some small constant $\rho > 0$. We will set $\rho$ later and we will assume $1/ \rho$ is an integer. 
  Let $|I|$ denote the length of the interval $I$; more precisely, the number of integer time steps in $I$. 
  The following properties of $\cI$ are immediate from the construction.

\begin{observation}
    \label{o:rt-1}
The collection $\cI$ of intervals has the following properties. Let $I_{H'}$ be the first interval of length greater than 1.
\begin{enumerate}
    \item For all $h \geq 1$, $|I_h | \leq |I_{h+1}|$.
    \item All intervals $I_1, I_2, \ldots I_{H'-1}$ have length 1.
    \item For all $h \geq H'$, $|I_h| \geq \frac{1}{1+2\rho} |I_{h+1}|$.
\end{enumerate}    
\end{observation}

In the residual optimum, we will pretend that time steps in each interval in $\cI$ are all identical. Towards this end, define $g(t)$ to be $a_h$ where $I_h := [a_h, a_{h+1})$ is the unique interval in $\cI$ including $t$. So, we will consider the following LP.
\begin{align}
\mbox{Approximate Residual (Time-Indexed) LP $\cL^{\cP}_\bx$:} \quad \quad    
     f(\bx):= \min \sum_{j \in A(t)} & \frac{w_j}{p_j} \sum_{\tilde t \geq 1} g(\tilde t) \cdot z_j(\tilde t) \label{lp:rt}\\
     \sum_{\tilde t \geq 1} z_j(\tilde t) \dd \tilde t &= x_j \quad  \forall j \in A(t) \nonumber \\
     \bz(\tilde t) &\in \cP \quad  \forall \tilde t \geq 1 \nonumber 
 \end{align}

Note that in reality, $g(\tilde t)$ is the same for all $\tilde t \in I_h$, so we can aggregate the variables $\{z_j(\tilde t)\}_{\tilde t \in I_h}$ into one; this approximate LP is of compact size. However, we take this full expanded view for the sake of analysis, assuming wlog that $z_j(\tilde t)$ is the same for all $\tilde t$ in each $I_h$, for any fixed job $j$.

It is easy to see that replacing $\tilde t$ with $g(\tilde t)$ does not change the residual objective's supermodularity and monotonicity. The only properties we need to check is that the residual objective's change rate remains the same within $1+\eps$ factor, under GD and the adversary, which can be easily offset by extra $1+O(\eps)$-speed augmentation. So, henceforth we aim to show it. 

\paragraph{GD remains effective in decreasing the residual optimum.}
We first show that processing following the optimum solution to the above LP in the first time step decreases the objective by an amount that is almost equal to the total fractional weight of all jobs. Let $\bz := \langle \bz_1, \bz_2, \ldots, \rangle$ be the optimum solution to $\cL_\bx$. We may use inequalities in the subscript to consider a partial schedule of $\bz$; for example, $\bz_{\geq 2} :=\langle \bz_2, \bz_3, \ldots, \rangle  $.

\begin{lemma}
    \label{lem:rt-1}
    Let $\bz := \langle \bz_1, \bz_2, \ldots, \rangle$ be the optimum solution to $\cL_\bx$. Then, 
    $f(\bx) - f(\bx - \bz_1) \geq \cL_\bx(\bz_{\geq 1})  - \cL_{\bx - \bz_1}(\bz_{\geq 2}) \geq \frac{1}{1+2\rho}\tilde W(t)$.
\end{lemma}

Before beginning proving the lemma, we first make a simple observation, which will be useful throughout this section. Intuitively, one wants to complete more fractional weights earlier to minimize the objective, which is formally stated below. 

\begin{claim}
    \label{c:rt-2}
    Let $V(\tau) :=\sum_j \frac{w_j}{p_j} z_j(\tau)$. We have $V(a_h) \geq V(a_{h+1})$.
    
\end{claim}
\begin{proof}
    Otherwise, swapping $\bz(a_h)$ and $\bz(a_{h+1})$ decreases the LP objective, contradicting $\bz$'s  optimality.
\end{proof}

\begin{proof}[\mbox{Proof of  Lemma~\ref{lem:rt-1}}]
    The first inequality, $f(\bx) - f(\bx - \bz_1) \geq \cL_\bx(\bz_{\geq 1})  - \cL_{\bx - \bz_1}(\bz_{\geq 2})$, follows from the observation that $\bz_{\geq 2}$ is a feasible schedule to $\cL_{\bx - \bz_1}$. Let $I_H$ be the last interval where $z_j(a_H)$ has a non-zero value for some $j$. 
    We then have,
    \begin{align*}
    &\cL_\bx(\bz_{\geq 1})  - \cL_{\bx - \bz_1}(\bz_{\geq 2}) \\
    = &\sum_j \frac{w_j}{p_j} \sum_{h \in [H]} (a_h - a_{h-1}) z_j(a_h) 
    = \sum_j \frac{w_j}{p_j} \sum_{h \in [H]} |I_{h-1}| z_j(a_h) \\
    \geq &\sum_j \frac{w_j}{p_j} \sum_{h \leq H'} |I_{h-1}| z_j(a_h) + \sum_j \frac{w_j}{p_j} \sum_{H' < h  \leq H} |I_{h-1}| z_j(a_h)     \\
    \geq &\sum_j  \sum_{h \leq H'} \frac{w_j}{p_j}  z_j(a_h) + \sum_j \frac{w_j}{p_j} \sum_{H' < h  < H} \frac{1}{1+2\rho}|I_{h}| z_j(a_h) \quad \mbox{[Observation~\ref{o:rt-1}]}\\
    = & \sum_{\tau < a_{H'}} V(\tau) + \frac{1}{|I_{H'}|} \sum_{a_{H'} \leq \tau < a_{H'+1}} V(\tau) + \frac{1}{1+2\rho} \sum_{\tau \geq a_{H'+1}} V(\tau) \\
    \geq &\frac{1}{1+2\rho} \sum_\tau  V(\tau)  \\
    = &\frac{1}{1+2\rho} \tilde W(t)  
    \end{align*}    
The penultimate inequality follows from Claim ~\ref{c:rt-2} and $|I_{H'}| \leq 2\rho a_{H'}$.
\end{proof}
By setting $\rho$ small enough (and using slightly more speed augmentation), we will argue that GD effectively decreases the residual optimum. 

\paragraph{No need to recompute the residual optimum until a new job arrives.}
In the above, we only showed the GD effectively decreases the residual objective for one time step. Here, we further argue that we can just process $\bz_1, \bz_2, ...., \bz_{\tilde t}$ until a new job arrives at time $t+ \tilde t$. To this end, we \emph{pretend} that the residual changes by 
\begin{equation}
    \label{eqn:pretend-credit}
\cL_{\bx - \bz_{\leq \tau -1}}(\bz_{\geq \tau})  - \cL_{\bx - \bz_{ \leq \tau}}(\bz_{\geq \tau + 1})
\end{equation}
when we process $z_\tau$ at time $t+\tau$. 
This is because until a new job arrives at time $\tau'$, the residual decreases by at least
$\cL_{\bx - \bz_{\leq 0}}(\bz_{\geq 1})  - \cL_{\bx - \bz_{ \leq \tau'}}(\bz_{\geq \tau' + 1})$, and we can decompose it as a telescopic sum, where we pretend that the residual decreases by the amount shown in   Eqn.~(\ref{eqn:pretend-credit}) at each time $\tau$.

\paragraph{No algorithm can decrease the residual optimum much more than GD.} We need to argue that no algorithm can decrease the residual optimum by more than $\tilde W(t)$ (with speed 1). To see this, fix an algorithm and suppose it processes $\bz_1$ in the first time slot, and the optimum schedule for the remaining jobs sizes, $\bx - \bz_1$, is $\langle \bz_2, \bz_3, \ldots, \rangle$. Our goal is to upper bound,
\begin{align*}
f(\bx) - \cL_{\bx - \bz_1} (\bz_{\geq 2})
\leq &\cL_{\bx - \bz_0} (\bz_{\geq 1})  -  \cL_{\bx - \bz_1} (\bz_{\geq 2}) \quad  \mbox{[$\bz_{\geq 1}$ is a feasible schedule to $\cL_\bx$ and $\bz_0 = 0$]}  \\
 \leq &\sum_j \frac{w_j}{p_j} \sum_{h \in [H]} (a_h - a_{h-1}) z_j(a_h) 
    = \sum_j \frac{w_j}{p_j} \sum_{h \in [H]} |I_{h-1}| z_j(a_h) \\
    \leq &\sum_j \frac{w_j}{p_j} \sum_{h \in [H]} |I_{h}| z_j(a_h) 
 \quad \mbox{[Observation~\ref{o:rt-1}]}\\    
    = & \sum_\tau  V(\tau) 
    =  \tilde W(t)        
    \end{align*}    

\begin{lemma}
    \label{lem:rt-3}
    No algorithm can decrease the residual optimum by more than $\tilde W(t)$ at a time $t$. 
\end{lemma}

\paragraph{Putting the Pieces Together.} We use the same potential function we used in the proof of Theorem~\ref{thm:procesing-rate-view}, except that $f$ is replaced with the approximate function we consider in Eqn. (\ref{lp:rt}). All the other properties such as supermodularity remain to be satisfied with the approximate residual LP. Thus, the former analysis for discrete events such as jobs arrival or completion continues to hold true (Lemmas~\ref{lem:second-thm-super}
and \ref{lem:discrete-change-processing-view}). 
The bound in Lemma~\ref{lem:second-theorem-gd} becomes 
$\frac{\dd}{\dd t} f(\bp^A(t)) = - \frac{s}{1+2\rho} \tilde W^A(t)$ thanks to Lemma~\ref{lem:rt-1}. The bound in Lemma~\ref{lem:second-theorem-adv} remains unchanged due to Lemma~\ref{lem:rt-3}. Then, assuming the online algorithm is given $(1+\eps)(1+2\rho)$-speed, the bound in Corollary~\ref{cor:running-clairvoyant} becomes
$\frac{d}{dt} \Phi(t) \leq -\tilde W^A(t) + 
\frac{1}{\eps}((1+\eps)(1+2\rho) + 2) \tilde W^O(t)$.
We can set $\rho$ to $\eps$ to obtain an algorithm that is $(1+\eps)(1+2\eps)$-speed $O(1/ \eps)$-competitive for the fractional objective. By scaling $\eps$ appropriately, we have Theorem~\ref{thm:procesing-rate-view}, even with the approximate residual optimum. 
Thus, we have shown polynomial time algorithms achieving the theorem---more precisely, assuming that we have a poly-time separation oracle for whether $\bz \in \cP$ or not. As a result, We also have polynomial time algorithms for Theorem~\ref{thm:resource-view-2} as the proof of the theorem reduces the resource view to the processing rate view.

\bibliography{ref}

\bibliographystyle{alpha}

\end{document}